\documentclass[10pt]{amsart}
\usepackage[square,sort,comma]{natbib}

\usepackage[margin=1.5in]{geometry}

\usepackage{concmath}

\usepackage[T1]{fontenc}
\usepackage[latin9]{inputenc}
\usepackage{varioref}
\usepackage{amsthm}
\usepackage{amsmath}
\usepackage{amssymb}
\usepackage{enumerate}
\usepackage[shortlabels]{enumitem}
\usepackage{tikz}
\usetikzlibrary{shapes,arrows}
\tikzstyle{decision} = [diamond, draw, fill=black!15, 
    text width=4.5em, text badly centered, node distance=3cm, inner sep=0pt]
\tikzstyle{block} = [rectangle, draw, fill=black!15, 
    text width=5em, text centered, rounded corners, minimum height=4em]
\tikzstyle{line} = [draw, -latex']
\tikzstyle{cloud} = [draw, ellipse,fill=red!20, node distance=3cm,
    minimum height=2em]

\usepackage{adjustbox}

\makeatletter
\theoremstyle{plain}
\theoremstyle{plain}
\newtheorem{thm}{Theorem}
  \theoremstyle{plain}
  \newtheorem{lem}[thm]{Lemma}
  \theoremstyle{plain}
  \newtheorem{defc}[thm]{Definition}
  \theoremstyle{plain}
  \newtheorem{prop}[thm]{Proposition}
  \theoremstyle{remark}
  \newtheorem*{note*}{Note}
  \theoremstyle{remark}
  \newtheorem*{conclusion*}{Conclusion}
  \theoremstyle{remark}
  \newtheorem{note}[thm]{Note}
 \theoremstyle{definition}
  \newtheorem{example}[thm]{Example}
  \theoremstyle{plain}

\newtheorem{theorem}{Theorem}[section]

\newtheorem{proposition}[theorem]{Proposition}

\newtheorem{lemma}[theorem]{Lemma}

\usepackage{amsthm}
\usepackage{amsfonts}
\usepackage{amscd}
\usepackage[cmtip,arrow]{xy}
\usepackage{pb-diagram,pb-xy}
\usepackage{youngtab}

\newcommand{\abs}[1]{\lvert#1\rvert}

\newcommand{\dif}{\textrm{\textbf{d}}}

\newcommand{\Img}[1]{\,\,\text{Im}\,#1}

\newcommand{\cD}{{\mathcal D}}

\newcommand{\cI}{{\mathcal I}}
\newcommand{\cJ}{{\mathcal J}}
\newcommand{\cL}{{\mathcal L}}

\newcommand{\mR}{\mathbb{R}}

\newcommand{\Implies}[2]{$\text{\ref{#1}}\implies\text{\ref{#2}}$}% X => Y

\newcommand{\pushright}[1]{\ifmeasuring@#1\else\omit\hfill$\displaystyle#1$\fi\ignorespaces}

\usepackage{verbatim}
\renewenvironment{reduce}{\verbatim}{\endverbatim}

\newcounter{VerbatimLineNo}
\newif\ifVerbatimLineNos
\VerbatimLineNostrue
\let\orig@verbatim@processline\verbatim@processline
\def\verbatim@processline{\ifVerbatimLineNos
  \stepcounter{VerbatimLineNo}\leavevmode 
  \makebox[1.3em][r]{\rm\footnotesize\theVerbatimLineNo}\ %
  \the\verbatim@line\par%
  \else\orig@verbatim@processline%
  \fi}

\def\StopVerbatimLineNos{\VerbatimLineNosfalse}
\def\ResumeVerbatimLineNos{\VerbatimLineNostrue}

\makeatother

\usepackage[english]{babel}

\begin{document}

\title[Cartan algorithm and Dirac constraints]{Cartan algorithm and Dirac constraints for Griffiths variational problems}

\author{H. Cendra}\email{hcendra@uns.edu.ar}
\author{S. Capriotti}\email{santiago.capriotti@uns.edu.ar}

%%%%%% For revtex4
%\affiliation{Departamento de Matem\'atica, \mbox{Universidad Nacional del Sur}, 8000 Bah\'{\i}a Blanca, Argentina.}

%%%%%% For amsart
\address{Departamento de Matem\'atica, \mbox{Universidad Nacional del Sur}, 8000 Bah\'{\i}a Blanca, Argentina.}

\maketitle

\begin{abstract}
  Dirac algorithm allows to construct Hamiltonian systems for singular systems, and so contributing to its successful quantization. A drawback of this method is that the resulting quantized theory does not have manifest Lorentz invariance. This motivated the quest of alternatives to the usual Hamiltonian theory on the space of sections; a particular instance of this search is the so called \emph{restricted Hamiltonian theory}, where the equations of motion of a field theory are formulated by means of a \emph{multisymplectic structure}, a kind of generalization of the symplectic form to the multidimensional context, and even a constraint algorithm working in this context has been proposed. In the present article we will try to provide partial aswers to two questions intimately related with these issues: First, to assign multisymplectic spaces to variational problems in the Griffiths formalism in such a way that the equations of motion can be written as restricted Hamiltonian systems, and second, to propose a covariant Dirac-like algorithm suitable to work with them; it must be recalled that given the Griffiths formalism contains the classical variational problems as particular instances, it yields to a novel covariant algorithm for deal with constraints in field theory. Moreover, in this formulation the constraint algorithm becomes simply the Cartan algorithm designed for deal with Pfaffian systems.
\end{abstract}

\tableofcontents

\section{Introduction}

\subsection{On the search of Hamiltonian structures for variational problems}

The Dirac algorithm is a useful tool in classical mechanics, allowing to find Hamiltonian descriptions for singular Lagrangians. It is a procedure aiming to locate a subset of the phase space, which will be called \emph{final constraint set}, characterized by the property that every of its points is contained in some solution of the system. This set can be obtained from the so called \emph{Gotay algorithm} \cite{GotayNester,cendra08:_geomet_const_algor_dirac_manif_applic}.

In the search of an analogous procedure for field theory, we first need to confront a fundamental issue: While there are essentially just only one possible Hamiltonian theory in classical mechanics, it is not the case in field theory, where multiple Hamiltonian schemes have been designed \cite{helein01hamiltonian,EcheverriaEnriquez:2005ht,Gunther198723,Echeverria-EnrIquez20007402,Munteanu20041730,Gotay1991203,KrupkaVariational}. Once a framework is fixed, there are a number of proposals on the way to proceed when the equations show some singular behaviour; we will explore in some depth just two approaches to the theme, those described in the references \cite{zbMATH02233555} and \cite{Seiler95involutionand}. These works proceed in different directions: The former uses restricted (pre)multisymplectic formalism in order to set a problem resembling the initial setting of Gotay algorithm, and then divides the algorithm in two steps, the first attacking the problem of \emph{tangency} of the solutions, and the second solving \emph{integrability} issues. The latter apply formal theory of PDE to the PDE system underlying Euler-Lagrange equations; although it requires to introduce local coordinates and so seems to be a less geometrical approach than the former, it has the advantage of dealing directly with integrability matters.

In order to remedy in some extent the difficulties found in the formal PDE approach, we will try to represent the underlying PDE system in terms of an object of geometrical nature, namely, by working in the realm of \emph{exterior differential systems} (EDS from now on) \cite{BCG}, which are ideals in the exterior algebra of a manifold closed under exterior differentiation. These geometrical objects can be used to encode in a geometrical fashion any PDE system, and have at the same time two crucial properties, namely they can be easily ``restricted'' to a submanifold, and there exists a way to ensures the integrability of the underlying PDE system, through the concept of \emph{involutivity}. It results that not every EDS has this property; nevertheless, there exists a procedure, called \emph{Cartan algorithm} \cite{CartanBeginners}, allowing us to construct an involutive EDS from any EDS of a particular kind, the so called \emph{linear Pfaffian EDS}.

It is our purpose here to make some contributions along these lines, by showing how to use Cartan algorithm, mainly designed for dealing with linear Pfaffian EDS, as a kind of Gotay algorithm, namely, as a way to bring into light the hidden constraints of such theories. It will draw upon some aspects of each of the approaches mentioned above, because it will employ the restricted Hamiltonian formalism, but making use at the same time of the tools available in theory of exterior differential systems in order to tackle the tangency and integrability issues in a quite natural and unified framework. It is interesting to note that a similar approach, in case of classical mechanics, was studied in \cite{Shadwicketal}; we will repeat the analysis carried out in this article within our framework (see Section \ref{Sec:ClassicalExample}) in order to be able to make comparisons between the procedures.

Additionally, it is important to make clear the range of problems where this algorithm could be applied; from a general viewpoint, the basic structure we need in order to be able to work with this procedure is a bundle $W$ on a manifold $M$ of dimension $m$ together with a closed $m+1$-form, namely assuming $W$ is a (pre)multisymplectic manifold of order $m+1$ \cite{JAZ:4974756}. This scope allows us to use the algorithm on variational problems of more general nature than those encountered in field theory; we are referring here to the variational problems in the so called \emph{Griffiths formalism} \cite{book:852048,hsu92:_calcul_variat_griff}. By using results of \cite{GotayCartan} and \cite{W-1982}, we will see that for those variational problem of this kind for which the \emph{canonical Lepage-equivalent variational problem} is contravariant, it is possible to define restricted Hamiltonian-like systems\footnote{In fact, an scheme for \emph{extended} Hamiltonian systems can be developed either, see Appendix \ref{App:EXtendedHamSys}.} for its underlying variational principle. After that, the scheme will be completed, and the algorithm could be applied. A consequence of the wider scope of our constructions is that we will be able to construct restricted Hamiltonian systems for PDE systems not directly related to classical variational problems, see for example Section \ref{Sec:EDSStrong}. In these cases the outcome of Cartan algorithm are the integrability conditions for the PDE system, as expected; thus we are dealing with an scheme that puts on equal footing Dirac-like constraints and integrability conditions of PDE systems.

So let us describe briefly the structure of this article: In Section \ref{Sec:MultipleVersions} we will discuss some equivalent formulations of the equations of motion associated to a Lepage-equivalent problem; the main result is the restricted Hamilton system form for these equations, valid in the general framework of Griffiths variational problems (see item \ref{statement4} in Proposition \ref{Prop:OursEquivalentOnVars}.) Section \ref{Sec:HamiltonianVersion} deals with the construction of a particular kind of Lepage-equivalent problem, carried out in first place for classical variational problems and generalized later for variational problems of broader nature (viz.\ Definition \ref{Def:RestHamSsytForVarsProb}.) The connection with Cartan algorithm is made into Section \ref{Sec:GotayAlgorithm}; the central object is the so called \emph{Hamilton submanifold}, which is a submanifold in a Grassmann bundle representing the restricted Hamiltonian system. A crucial feature of this submanifold is that it gives rise to a linear Pfaffian system through the pullback of the canonical contact structure, suitable for the application of the Cartan algorithm and thus producing a subset where involutivity can be achieved. Finally, Section \ref{Sec:Examples} deals with the applications of these ideas to concrete examples. They cover diverse topics, showing the versatility of the scheme; concretely, we carried out with enough detail the procedure in case of variational problems associated to classical mechanics, first order field theory and PDE systems with integrability conditions. It is important to point out that the whole calculations involved in these examples were performed by using the package \emph{EDS} \cite{IOPORT.01307045} of the computer algebra system \emph{Reduce} \cite{Hearn:1967:RUM}; some indications on the actual implementation are discussed in the example dealing with Maxwell equations.

\section{Multiple versions for the variational equations}\label{Sec:MultipleVersions}

% Let us work in the setting of the following diagram
% \[
% \begin{diagram}
%   \node{\bigwedge^mE}\arrow{e,t}{\kappa}\arrow{se,b}{\overline\kappa}\node{E}\arrow{s,r}{\pi}\\
%   \node[2]{M}
% \end{diagram}
% \]
% We will begin with a variational problem $\left(E,\lambda,\cI\right)$ such that $\cI$ is generated by the sections of a graded subbundle $I\subset\bigwedge^\bullet E$; therefore we can define (following \cite{GotayCartan}) a subbundle $W_\lambda\stackrel{\imath}{\hooklongrightarrow}\bigwedge^mE$ via
% \[
% \left.W_\lambda\right|_e:=\left\{\lambda\left(e\right)+\left.I\right|_e\cap\wedge{}^m_eE\right\}.
% \]
% By mimicking the constructions found in the literature (for example, Thm. 5 in page 17 of \cite{EcheverriaEnriquez:2005ht}), we need to find a closed $1$-form $\alpha$ on $\bigwedge^mE$ such that for every $X_\alpha\in T^mE$ solution of
% \[
% X_\alpha\lrcorner\Omega=\left(-1\right)^{m+1}\alpha,\qquad X_\alpha\lrcorner\overline\kappa^*\eta=1,
% \]
% where $\eta$ is a volume form on $M$, there exists $X_\imath\in T^mW_\lambda$ such that $T\imath\left(X_\imath\right)=X_\alpha$ and
% \[
% X_\imath\lrcorner\dif\lambda=0,\qquad X_\imath\lrcorner\left(\overline\kappa\circ\imath\right)^*\eta=1.
% \]
% Some additional coditions (for example, $\text{Im}\imath$ an integral submanifold for $\alpha^\circ$) could be necessary in this scheme; nevertheless, it seems to have the expected features of the generalized scheme.

\subsection{Geometrical preliminaries}

\subsubsection{Variational problems}

It will be necessary to introduce the basic language we will use in the rest of the article. The first concept we will introduce is a generalization of regular distributions on a manifold.
\begin{defc}[Exterior differential system]
  An \emph{exterior differential system} on a manifold $M$ is an ideal in its exterior algebra\footnote{Recall that the \emph{exterior algebra of a manifold $M$} is the algebra of sections of the bundle $\wedge^\bullet M\rightarrow M$.}, closed by exterior differentiation.
\end{defc}
We will assume the reader knows the basic facts related to these geometrical objects; the standard references are \cite{BCG,CartanBeginners}, and \cite{nkamran2000} can be found helpful. Our next task is to set what a variational problem is in this context; the following definition is extracted from \cite{GotayCartan}.
\begin{defc}[Variational problem]
  Given a triple $\left(F\stackrel{\pi_1}{\longrightarrow}M,\lambda,\cI\right)$ where $F\rightarrow M$ is a bundle on a manifold of dimension $m$, $\lambda\in\Omega^m\left(F\right)$ and $\cI$ is an EDS on $F$, the associated \emph{variational problem} consists in finding the extremals of the map
  \[
  \sigma\mapsto\int_M\sigma^*\lambda
  \]
  for $\sigma$ living in the set of integral sections of the EDS $\cI$ for the bundle $F$.
\end{defc}
\begin{example}[Classical variational problem]
  The classical variational problem is the variational problem (in the sense of the previous definition) equivalent to the first order field theory, consisting into the choices $F:=J^1\pi$ for the bundle of fields $\pi:E\rightarrow M$,
\[
\cI:=\cI_{\text{con}}=\text{ contact structure of the jet space}
\]
and $\lambda:=L\eta$, where $L$ is a function on $J^1\pi$ and $\eta$ is a volume form.
\end{example}
In the search of equations for a variational problem, there are some subtleties to be taken into account (see \cite[p.\ 40]{KrupkaVariational}); nevertheless, there exists a way to found an EDS format of the equations of motion, through the concept of Lepage-equivalence.
\begin{defc}[Lepage-equivalent problem]\label{Def:LepEquivProblem}
  A variational problem of the form
  \[
  \left(W\stackrel{\rho}{\longrightarrow}M,\alpha,0\right)
  \]
  is a \emph{Lepage-equivalent problem} for the variational problem $\left(F\stackrel{\pi_1}{\longrightarrow}M,\lambda,\cI\right)$ if there exists a surjective submersion $\nu:W\rightarrow F$ such that
  \begin{enumerate}[(i)]
  \item The following diagram is commutative
    \[
    \begin{diagram}
      \node{W}\arrow{e,t}{\nu}\arrow{se,b}{\rho}\node{F}\arrow{s,r}{\pi_1}\\
      \node[2]{M}
    \end{diagram}
    \]
  \item For every section $\gamma:M\rightarrow W$ of $\rho$ such that $\nu\circ\gamma$ is an integral section of $\cI$, we have that
    \[
    \gamma^*\alpha=\left(\nu\circ\gamma\right)^*\lambda.
    \]
  \end{enumerate}
\end{defc}
The first thing we need to know is that the equations for extremals of a Lepage-equivalent problem \emph{are easy to obtain}; it will be proved in Proposition \ref{Prop:OursEquivalentOnVars} below. The question is whether the set of extremals is preserved in some way when we change a variational problem with one of its Lepage-equivalents; in general it is not true, so it is necessary to introduce the following terminology.
\begin{defc}[Covariant and contravariant Lepage-equivalent problems]
  A Lepage-equiva\-lent problem is \emph{covariant} if the projection along $\nu$ of every of its extremals is an extremal of the original variational problem; in the same vein, we say that it is \emph{contravariant} if every extremal of the original problem can be lift through $\nu$ to an extremal of the Lepage-equivalent.
\end{defc}
The following result can be found in \cite{GotayCartan}.
\begin{prop}
  The classical variational problem admits a Lepage-equivalent problem co- and contravariant.
\end{prop}
This proposition is proved there by using an standard construction known as \emph{classical Lepage-equivalent problem}; in Section \ref{Sec:HamiltonianVersion} we will adapt this construction in order to found a restricted Hamiltonian system for a variational problem.

\subsubsection{Several descriptions for EDS}

Before to continue with the article, it is necessary to introduce some additional terminology, intended to reduce ambiguities in the discussions we will perform below. The most easy way to generate an ideal in $\Omega^\bullet\left(M\right)$ is by means of a set of forms.
\begin{defc}[EDS generated by a set of forms]
  Let $S:=\left\{\alpha_i\right\}_{i\in I}\subset\Omega^\bullet\left(M\right)$ be a set of forms on $M$. The EDS \emph{generated by $S$} is the minimal ideal (respect to inclusion) closed by exterior differentiation containing the set $S$.
\end{defc}
There exists a more or less explicit description for the EDS generated by a \emph{finite} set of forms.
\begin{prop}
  Let $S_0:=\left\{\alpha_1,\cdots,\alpha_s\right\}$ be a finite set of forms. Then the EDS $\cI_{S_0}$ generated by $S_0$ is the set
  \[
  \cI_{S_0}:=\left\{\beta^1\wedge\alpha_1+\cdots+\beta^s\wedge\alpha_s+\gamma^1\wedge\dif\alpha_1+\cdots+\gamma^s\wedge\dif\alpha_s:\beta^i,\gamma^j\in\Omega^\bullet\left(M\right)\right\}.
  \]
\end{prop}
There exists another way to generate EDS, which will be very important in the present work, namely, by using local sections of a subbundle of forms. The next definition, which has been adapted from \cite[Prop.\ $2.28\left(b\right)$]{MR0295244}, describes this kind of EDS.
\begin{defc}[EDS generated by sections]\label{Def:EDSGenerated}
  Let $\cI\subset\Omega^\bullet M$ be an EDS and $I\subset\wedge^\bullet\left(M\right)$ a subbundle of the bundle of forms on $M$. We will say that $\cI$ is \emph{locally generated by the sections of $I$} if and only if there exists an open cover $\left\{U\right\}$ of $M$ such that for every open set $U$ in the cover
  \begin{enumerate}[(i)]
  \item any section $\sigma:U\rightarrow I$ is an element of $\left.\cI\right|_U$, and
  \item for every $\beta\in\cI$ we can find a finite collection of local sections $\sigma_1,\cdots,\sigma_s:U\rightarrow I$ and functions $f_1,\cdots,f_s,g_1,\cdots,g_s$ on $U$ such that
  \[
  \left.\beta\right|_U=f_1\sigma_1+\cdots+f_s\sigma_s+g_1\dif\sigma_1+\cdots+g_s\dif\sigma_s.
  \]
  \end{enumerate}
\end{defc}
The first condition guarantees minimality of the rank of the subbundle $I$; without it, nothing prevents us to take the total space $\wedge^\bullet M$ as subbundle fulfilling the second requirement in the definition.
\begin{example}[The contact structure is generated by sections]\label{Ex:ContactGenerated}
  Let us recall that for every adapted coordinate chart $U$, the EDS $\left.\cI_{\text{con}}\right|_U$ contains the forms
  \[
  \theta^A:=\dif u^A-u^A_k\dif x^k;
  \]
  in fact, this EDS is differentially generated by them. Thus we can define the subbundle $I_c\subset\wedge^\bullet\left(J^1\pi\right)$ whose fibers are given by
  \[
  \left.I_c\right|_{j_x^1s}:=\left\{\sum_A\alpha_A\wedge\left.\theta^A\right|_{j_x^1s}:\alpha_A\in\wedge^\bullet_{j_x^1s}\left(J^1\pi\right)\right\};
  \]
  the local sections of $I_c$ generate $\cI_{\text{con}}$ in the sense of Definition \ref{Def:EDSGenerated}.
\end{example}

\subsection{Multiple versions of the variational equations}
In this section we will describe how to construct a (pre)multisymplectic manifold for every variational problem with a co- and contravariant Lepage-equivalent variational problem. In order to achieve this goal, it will be important to review the way in which such construction is carried out in classical field theory. So let $\pi:E\rightarrow M$ be a fibration, with $M$ a compact manifold with $\partial M=\emptyset$. Then we have the following bundles
\[
\begin{diagram}
  \node{J^1\pi}\arrow{se,b}{\pi_{10}}\node{\wedge{}^m_2E=:\mathcal{M}\pi}\arrow{s,l}{}\arrow{e,t}{\mu}\node{J^1\pi^*:=\wedge{}^m_2E/\wedge{}^m_1E}\arrow{sw,b}{\bar\tau}\\
  \node[2]{E}\arrow{s,l}{}\\
  \node[2]{M}
\end{diagram}
\]
The space $\mathcal{M}\pi\subset\wedge{^mE}$ is a multisymplectic space with an $m+1$-form $\Omega$, by using the differential of the restriction $\Theta$ of the canonical $m$-form on $\wedge{}^mE$.

\begin{defc}[Restricted Hamiltonian system - provisional definition]
  The \emph{restricted Hamiltonian system} in the classical setting is the couple $\left(J^1\pi^*,h\right)$ where $h:J^1\pi^*\rightarrow\mathcal{M}\pi$ is a section of $\mu$, known as \emph{Hamiltonian section}.
\end{defc}
We will take this definition as provisional, because we will replace it below with another, more general (see Definition \ref{Def:NewDefRestHamSystem}), suitable to work with the general variational problems we will look into in this work.

In classical first order field theory given by $\left(J^1\pi\stackrel{\pi}{\rightarrow}M,L\eta,\cI_{\text{con}}\right)$, there exists a procedure \cite{de1996geometrical,2011arXiv1110.4778C} allowing us (under mild conditions) to find both \emph{extended} and \emph{restricted} Hamiltonian systems on $\mathcal{M}\pi$ and $J^1\pi^*$ respectively, whose solutions can be put into correspondence; we refer to this reference for details. Briefly, given a \emph{hyperregular} Lagrangian $L$, it can be defined a Hamiltonian section $h:J^1\pi^*\rightarrow\mathcal{M}\pi$ given by
\[
h:\left(x^i,u^A,p_B^k\right)\mapsto\left(x^i,u^A,p_B^k,u_k^Bp_B^k-L\right)
\]
such that there exists a one-to-one correspondence between the solutions of the Euler-Lagrange equations for $L$ and the solutions of the restricted Hamiltonian system $\left(J^i\pi^*,h\right)$; in this formula it is considered that $u^A_k$ is a function on $J^1\pi^*$ defined through the equation
\[
p_A^k=\frac{\partial L}{\partial u_k^A}.
\]
Now, there exists several ways to formulate the equations associated to the restricted Hamiltonian system \cite{EcheverriaEnriquez:2005ht}:
\begin{itemize}
\item By seeking sections $\sigma:M\rightarrow J^1\pi^*$ which are extremals for the action integral
  \[
  \sigma\mapsto\int_M\sigma^*\Theta_h,
  \]
  where $h:J^1\pi^*\rightarrow\mathcal{M}\pi$ is a \emph{Hamiltonian section} of $\mu$ and
  \[
  \Theta_h:=h^*\Theta.
  \]
\item By finding $m$-dimensional integral submanifolds of the EDS
  \[
  \cI_{\text{HC}}:=\left\langle X\lrcorner\Omega_h:X\in\mathfrak{X}^{V\left(\bar\tau\right)}\left(J^1\pi^*\right)\right\rangle 
  \]
  where $\Omega_h:=h^*\Omega$, transverse to $\bar\tau$.
\item By discovering $m$-dimensional integral submanifolds of the EDS
  \[
  \cI_{\text{HC}}:=\left\langle X\lrcorner\Omega_h:X\in\mathfrak{X}\left(J^1\pi^*\right)\right\rangle ,
  \]
  transverse to $\bar\tau$.
\item By seeking integrable $m$-multivectors $X_h\in\Gamma\left(\wedge{}^mJ^1\pi^*\right)$ such that
  \[
  X_h\lrcorner\Omega_h=0.
  \]
\end{itemize}
When starting with a general variational problem $\left(F\rightarrow E\rightarrow M,\lambda,\cI\right)$ admitting a co- and contravariant Lepage equivalent variational problem $\left(W\rightarrow M,\alpha,0\right)$, the following result sets some equivalences analogous to the enumerated above.
\begin{prop}\label{Prop:OursEquivalentOnVars}
  Let $\left(W\stackrel{\tau}{\rightarrow}M,\alpha,0\right)$ be a variational problem. Then the following assertions are equivalent for a section $\sigma:M\rightarrow W$:
  \begin{enumerate}[label=(\arabic*),ref=(\arabic*)]\label{}
  \item $\sigma$ is an extremal for the variational problem
    \[
    \sigma\mapsto\int_M\sigma^*\alpha.
    \]\label{statement1}
  \item $\sigma$ is an integral section of the EDS
    \[
    \cI_{\text{HC}}:=\left\langle X\lrcorner\dif\alpha:X\in\mathfrak{X}^{V\left(\tau\right)}\left(W\right)\right\rangle .
    \]\label{statement2}
  \item $\sigma$ is an integral section of the EDS
    \[
    \overline{\cI}_{\text{HC}}:=\left\langle X\lrcorner\dif\alpha:X\in\mathfrak{X}\left(W\right)\right\rangle .
    \]\label{statement3}
  \item $\sigma$ is integral for a local decomposable $m$-multivector $Z_m\in\Gamma\left(\wedge{}^mW\right)$ such that
    \[
    Z_m\lrcorner\dif\alpha=0.
    \]\label{statement4}
  \end{enumerate}
\end{prop}
\begin{proof}
  We proceed separately:
  \begin{description}
  \item[\Implies{statement1}{statement2}] The variations of a section $\sigma$ is a section of the pullback bundle $\sigma^*\left(VW\right)$, where $VW\subset TW$ indicates the subbundle of vertical vectors; by extending one of these variations to a true vector field $\delta V$ on $W$, we obtain the following formula for the variation of the action functional
    \[
    \int_M\sigma^*\left(\cL_{\delta V}\alpha\right)=0.
    \]
    By Cartan's magic formula and using the fact that variations annihilates on the boudary of $M$ it follows~\ref{statement2}.
  \item[\Implies{statement2}{statement3}] Let us suppose that $\sigma$ verifies~\ref{statement2}, and let $w$ belongs to $\Img{\sigma}$. Then we have the decomposition
    \[
    T_wW=V_\sigma W\oplus T_w\sigma\left(T_{\tau\left(w\right)}M\right)
    \]
    and so $X=X^V+T_w\sigma\left(V\right)$ for every $X\in T_wW$; here $V\in T_{\tau\left(w\right)}M$. Contracting it with the form $\dif\alpha$ we obtain
    \[
    X\lrcorner\dif\alpha=T_w\sigma\left(V\right)\lrcorner\dif\alpha
    \]
    and so
    \[
    \sigma^*\dif\alpha=\sigma^*\left(T_w\sigma\left(V\right)\lrcorner\dif\alpha\right)=V\lrcorner\sigma^*\dif\alpha=0
    \]
    because $\dif\alpha$ has degree $m+1$ and $\Img\sigma$ has dimension $m$.  Thus~\ref{statement3} follows.
  \item[\Implies{statement3}{statement4}] Let $\sigma$ be a section fulfilling~\ref{statement3} and $U\subset M$ an open set whose tangent bundle is trivializable; then we can define $Z\in\Gamma\left(\wedge{}^mW\right)$ such that on $\Img{\sigma}$ takes the values
    \[
    Z\circ\sigma=\wedge_{i=1}^mT\sigma\left(V_i\right)
    \]
    where $\left\{V_1,\cdots,V_m\right\}$ is a basis on $TU\subset TM$. It could be done because $\Img\sigma$ is closed in $W$. Then if $X$ is an arbitrary vector field on $W$, we will have that
    \[
    0=\sigma^*\left(X\lrcorner\dif\alpha\right)\left(V_1,\cdots,V_m\right)=Z\lrcorner X\lrcorner\dif\alpha
    \]
    and ~\ref{statement4} follows when we realize that $X$ is arbitrary.
  \item[\Implies{statement4}{statement1}] Let $\delta V\in\mathfrak{X}\left(W\right)$ be an arbitrary vertical vector field; so we will have that
    \begin{align*}
      \delta V\lrcorner Z_m\lrcorner\dif\alpha&=\left(-1\right)^mT\sigma\left(\partial_1\wedge\cdots\wedge\partial_m\right)\lrcorner\left(\delta V\lrcorner\dif\alpha\right)\\
      &=\left(-1\right)^m\left(\partial_1\wedge\cdots\wedge\partial_m\right)\lrcorner\sigma^*\left(\delta V\lrcorner\dif\alpha\right).
    \end{align*}
    So from $Z_m\lrcorner\dif\alpha=0$ it results that $\sigma^*\left(\delta V\lrcorner\dif\alpha\right)=0$ for any vertical vector field $\delta V$, and thus
    \[
    \int_M\sigma^*\left(\cL_{\delta V}\alpha\right)=0
    \]
    if we suppose that $\left.\delta V\right|_{\partial M}=0$.$\qedhere$
  \end{description}
\end{proof}
This Proposition is giving us a representation for the equations that rule the extremals in terms of an EDS \cite{GoldSternberg}.
\begin{defc}[Hamilton-Cartan EDS]
  Given a variational problem $\left(W\rightarrow M,\alpha,0\right)$, the EDS $\cI_{\text{HC}}$ whose integral sections are exactly its extremals is called \emph{Hamilton-Cartan EDS}.
\end{defc}
\begin{note}[How to reduce to the restricted Hamiltonian system setting]
  The equivalences detailed above for restricted Hamiltonian systems could be obtained from this Proposition by setting $W=J^1\pi^*,\tau=\bar\tau$ and $\alpha:=\Theta_h$.
\end{note}
The Proposition \ref{Prop:OursEquivalentOnVars} means that in order to find a restricted Hamiltonian version of a variational problem $\left(F\rightarrow M,\lambda,\cI\right)$, it will be necessary to replace them by another variational problem $\left(W\rightarrow M,\alpha,0\right)$ with trivial restriction EDS. To the description of a possible replacement is devoted the next section.

%{\pigpenfont CECILIA VASCONI}

\section{On the Hamiltonian version of field theory}\label{Sec:HamiltonianVersion}

It is time now to study a method for the construction of a covariant Lepage-equivalent problem to every variational problem. This method is based in the works \cite{GotayCartan,W-1982,Gotay1991203}, although some modifications were introduced in order to adapt it to this context. By exploring a bit further the constructions yielding to a restricted Hamiltonian system in the classical case, we will obtain a hint on the key components of such a construction, and we will use these findings in order to set an analogous structure in variational problems of more general nature than those we encountered in field theory (see the examples in Section \ref{Sec:Examples}.)

\subsection{An scheme for first order field theory}

We will set up the restricted Hamiltonian version for first order field theory. Although it is a well-known scheme, we will present it in a non traditional fashion, highlighting those features which will prove to be important to generalize it to non classical variational problems. Our starting point is the classical variational problem
\[
\left(J^1\pi,L\eta,\cI_{\text{con}}\right),
\]
where $\pi:E\rightarrow M$ is the bundle of fields, $\eta$ is a volume form on $M$ and $\cI_{\text{con}}$ is the contact EDS on $J^1\pi$. With our purpose in mind, let us define
\begin{equation}\label{Eq:ClassicalVerticalSubBundle}
  \left(\wedge{}^m_2J^1\pi\right)^V:=\wedge{}^m_2J^1\pi\cap\left(V\pi_{10}\right)^0,
\end{equation}
where $\left(\cdot\right)^0$ indicates the annihilator of the vector bundle placed between the parenthesis. Moreover, let us consider the pullback bundle
\[
\begin{diagram}
  \node{\pi_{10}^*\left(\wedge{}^m_2E\right)}\arrow{e,t}{p_2}\arrow{s,l}{p_1}\node{\wedge{}^m_2E}\arrow{s,r}{\bar\tau^m_E}\\
  \node{J¹\pi}\arrow{e,b}{\pi_{10}}\node{E}
\end{diagram}
\]
\begin{lem}\label{Lem:BundleIsomorph}
  We have the bundle isomorphism
  \[
  \left(\wedge{}^m_2J^1\pi\right)^V\simeq\pi_{10}^*\left(\wedge{}^m_2E\right)
  \]
  as bundles on $J^1\pi$.
\end{lem}
\begin{note}[On the pullback]
  It is necessary to point out that $\pi_{10}^*\left(\wedge_2^m\pi\right)$ is nothing but the bundle $J^1\pi\oplus Z$, sometimes called \emph{Pontryagyn bundle}, as in \cite{Vankerschaver:1282812}; under this identification, the main point of the previous lemma is to show how such space adquires a multisymplectic structure. Moreover, the definition adopted here will allow us to translate the scheme to another kind of variational problems, as will be shown below.
\end{note}
Then we define the bundle $\widetilde{W}_{L\eta}\rightarrow J^1\pi$ such that
\[
\left.\widetilde{W}_{L\eta}\right|_{j_x^1s}:=\left(L\left(j^1_xs\right)\eta+\left.I_c\right|_{j_x^1s}\cap\wedge^mJ^1\pi\right)\cap\left(\wedge{}^m_2J^1_{j_x^1s}\pi\right)^V
\]
where $I_c$ is the bundle in $\wedge{}^\bullet J^1\pi$ whose sections generate $\cI_{\text{con}}$ (see Example \ref{Ex:ContactGenerated}); it is quite amusing to note that this definition uses only the data provided by the actual definition of the variational problem, and so it is suitable for the desired generalization.
\begin{defc}\label{Def:ClassLepEquiv}
  The subbundle
  \[
  \pi_{10}^*\left(\wedge{}^m_2E\right)\supset W_{L\eta}\stackrel{p_1}{\longrightarrow} J^1\pi
  \]
  which is the image of the set $\widetilde{W}_{L\eta}$ under the isomorphism of Lemma \ref{Lem:BundleIsomorph} will be called \emph{classical Lepage-equivalent} of the first order field theory $\left(J^1\pi,L\eta,\cI_{\text{con}}\right)$.
\end{defc}
From now on, in those places where no danger of confusion arise, we will use the same symbol to refer to $\widetilde{W}_{L\eta}$ and ${W}_{L\eta}$. 
\begin{note}[On terminology]
  The definition adopted above is based in the work of \cite{GotayCartan}; every variational problem is associated to a simplified variational problem, named \emph{canonical Lepage-equivalent}. It must be stressed that this definition is consistent with Definition \ref{Def:LepEquivProblem}, namely, that the classical Lepage-equivalent problem (in the sense of Definition \ref{Def:ClassLepEquiv}) is a Lepage-equivalent (in the sense of Definition \ref{Def:LepEquivProblem}) of the classical variational problem.
\end{note}
\begin{note}[On the choice of the bundles of forms]
  It is natural to ask here on the reason to choose $\wedge^m_2J^1\pi$ in the previous adopted definitions. It could be justified \emph{a posteriori} by looking at $\widetilde{W}_{L\eta}$: It is an object living naturally into $\wedge^mJ^1\pi$, as the original definition of Gotay shows, but we are embedding it into $\wedge^m_2J^1\pi$ in order to keep close to the known formalism for first order field theories. An hypothesis about the verticality order of the forms being considered as relevant in this formalism could be the following: It is the minimal order keeping the generators of the restriction EDS $\cI_{\text{con}}$ into $\widetilde{W}_{L\eta}$. Perhaps this remark would be taken into account whenever variational problems with higher order restriction EDS could arise.
\end{note}

Let $\Theta^V$ be the restriction of the canonical $m$-form on $\wedge{}^m_2J^1\pi$ to $\left(\wedge{}^m_2J^1\pi\right)^V$. The isomorphism found above allows us to define on $\pi_{10}^*\left(\wedge{}^m_2E\right)$ an $m$-form $\Theta_{10}$; let $\Theta_{L\eta}$ be the restriction of this $m$-form to $W_{L\eta}$. 
\begin{prop}\label{Prop:HamiltonCartanLeta}
  The Hamilton-Cartan EDS
  \[
  \cI_{\text{HC}}=\left\langle V\lrcorner\dif\Theta_{L\eta}:V\in TW_{L\eta}\right\rangle _{\text{diff}}
  \]
  on $W_{L\eta}$ is equivalent to the Euler-Lagrange equations associated to the variational problem $\left(J^1\pi,L\eta,\cI_{\text{con}}\right)$.
\end{prop}
\begin{proof}
  This can be proved in general grounds by using the co- and contravariance of the classical Lepage-equivalent problem; nevertheless, it is intructive to give a proof involving local coordinates. In fact, $\alpha\in W_{L\eta}$ iff in the adapted coordinates $\left(x^i,u^A,u^A_i,p_A^k,p\right)$ on $\pi_{10}^*\left(\wedge{}^m_2E\right)$
  \[
  \alpha=p_A^k\dif u^A\wedge\eta_k+\left(L-p_A^lu^A_l\right)\eta
  \]
  where $\eta_k:=\partial_k\lrcorner\eta$ (and we are assuming that $\dif\eta=0$); it means in particular that $W_{L\eta}$ can be described by the set of coordinates $\left(x^k,u^A,u^A_k,p_A^k\right)$ via
  \[
  W_{L\eta}=\left\{\left(x^i,u^A,u^A_i,p_A^k,L-p_A^lu^A_l\right)\right\}.
  \]
  Then
  \begin{equation*}
    \partial_i\lrcorner\dif\Theta_{L\eta}=\dif p_A^k\wedge\left(\dif u^A-u_l^A\dif x^l\right)\wedge\eta_{ik}+\left(\frac{\partial L}{\partial u^A_k}-p^k_A\right)\dif u^A_k\wedge\eta_i,
  \end{equation*}
  where $\eta_{ik}:=\partial_i\lrcorner\partial_k\lrcorner\eta$; but we have
  \begin{align*}
    \delta p_A^k\lrcorner\dif\Theta_{L\eta}&=\delta p_A^k\left(\dif u^A-u_l^A\dif x^l\right)\wedge\eta_{k}\\
    \delta u^A_k\lrcorner\dif\Theta_{L\eta}&=\delta u^A_k\left(\frac{\partial L}{\partial u^A_k}-p^k_A\right)\eta,
  \end{align*}
  so the contraction along elements of the form $\partial_i$ does not implies new generators for the EDS $\cI_{\text{H-C}}$.
\end{proof}
It is time to obtain the restricted Hamiltonian system from this scheme; in order to do that, it is necessary to work with a subbundle of $W_{L\eta}$. It can be seen as the subset generated by the set of zero forms belonging to the Hamilton-Cartan EDS when consider the independence condition $\eta\not=0$.
\begin{prop}
  The zero forms of the EDS $\cI_{\text{HC}}$ define a subset of $W_{L\eta}$ which can be generated as the image of a section $\sigma:J^1\pi\rightarrow W_{L\eta}$.
\end{prop}
\begin{proof}
  In fact, from the Corollary \ref{Prop:HamiltonCartanLeta}, the elements of the form
  \[
  \left(0,0,\delta u^A_i,0,0\right)
  \]
  gives rise to the desired zero forms according to the formula
  \[
  \delta u^A_k\left(\frac{\partial L}{\partial u^A_k}-p_A^k\right)\eta=0
  \]
  for all $\delta u^A_k$, once we realize that $\eta\not=0$ on integral sections. Thus the section $\sigma$ reads locally as follows
  \[
  \sigma\left(x^k,u^A,u^A_k\right)=\left(x^k,u^A,u^A_k,\frac{\partial L}{\partial u^A_l},L-\frac{\partial L}{\partial u^A_l}u^A_l\right)\in W_{L\eta}.\qedhere
  \]
\end{proof}
The section $\sigma$ is nothing but the Legendre transformation $\text{Leg}_L:J^1\pi\rightarrow\mathcal{M}\pi$. When Lagrangian $L$ is hyperregular, the restriction of the EDS $\cI_{\text{HC}}$ makes sense, because its integral sections can be put into one-to-one correspondence with solutions of the Euler-Lagrange equations for $L$. Nevertheless it loses some of its power when singular Lagrangian are considered, because further restrictions will be needed to achieve the desired correspondence; given that in this work we are searching for an algorithm allowing us to find these constraints, it will be no reason in restricting \emph{a priori} to $\text{Im}\left(\text{Leg}_L\right)$. These considerations will be taken into account in the following, when $W_{L\eta}$ will be assumed as the main object defining the restricted Hamiltonian systems associated to a variational problem, instead of the graph of the Hamiltonian section $h$ (or the Legendre transformation.) So we are ready to adopt a new definition for \emph{restricted Hamiltonian systems} suitable to be applied to the kind of variational problems we are considering here.
\begin{defc}[Restricted Hamiltonian system - new definition]\label{Def:NewDefRestHamSystem}
  A \emph{restricted Hamiltonian system} will be from now on a couple $\left(W,\Omega\right)$ consisting of a $m+1$-premultisymplectic manifold $W$ and its presymplectic form $\Omega$.
\end{defc}
By taking $W:=J^1\pi^*$ and $\Omega:=\Omega_h$ we obtain the restricted Hamiltonian systems in the old sense; the new definition will include some additional cases in which the Hamiltonian section is not easy to define (see the examples discussed in Section \ref{Sec:Examples}.) Nevertheless, as indicated above, whenever a restricted Hamiltonian system $\left(W,\Omega\right)$ is at our disposal, some sort of Hamiltonian equations can be formulated.
\begin{defc}[Hamilton equations associated to a restricted Hamiltonian system]
  The \emph{Hamilton equations} for the restricted Hamiltonian system $\left(W,\Omega\right)$ are the equations
  \[
  Z\lrcorner\Omega=0
  \]
  determining a decomposable $m$-vector field $Z\in\mathfrak{X}^m\left(W\right)$. Whenever a volumen form $\eta\in\Omega^m\left(W\right)$ is fixed, we can use the additional equation
  \[
  Z\lrcorner\eta=1
  \]
  as part of the Hamilton equations.
\end{defc}

% \begin{theorem}
%   The $m$-dimensional integral elements of $\cI_{\text{H-C}}$ on $W_{L\eta}$ coincide with the solutions of the restricted Hamiltonian system
%   \[
%   X\lrcorner\dif\Theta_{L\eta}=0,\qquad X\in T^mW_{L\eta}
%   \]
%   through the correspondence \textsf{decomposable $m$-vector $\rightarrow$ $m$-subspace}.
% \end{theorem}

\subsection{Hamiltonian structures for general variational problems}

Let us attack the problem of building a Hamiltonian structure for a general variational problem. The basic structure is a triple $\left(F\rightarrow E\rightarrow M,\lambda,\cI\right)$, where $\pi_{10}:F\rightarrow E$ is a fiber bundle on $E$, it is a bundle $\pi:E\rightarrow M$ on a $m$-dimensional manifold $M$, $\lambda$ is a $m$-form and $\cI$ is an EDS, both on $F$, which plays a rôle analogous to the jet space $J^1\pi$ in this setting. In these terms, we define
\[
\left(\wedge^m_pF\right)^V:=\wedge^m_pF\cap\left(V\pi_{10}\right)^0
\]
for some integer $p$ to be chosen according to a criteria we will set below, and the pullback bundle
\[
\begin{diagram}
  \node{\pi_{10}^*\left(\wedge^m_pE\right)}\arrow{e,t}{p_2}\arrow{s,l}{p_1}\node{\wedge^m_pE}\arrow{s,r}{\bar\tau^m_E}\\
  \node{F}\arrow{e,b}{\pi_{10}}\node{E}
\end{diagram}
\]
The bundle $\left(\wedge^m_pF\right)^V$ is the set of $p$-horizontal $m$-forms having ``no differential in the velocities direction''; as before, the following result holds.
\begin{lemma}
  There exists an isomorphism
  \[
  \left(\wedge^m_pF\right)^V\simeq\pi_{10}^*\left(\wedge^m_pE\right)
  \]
  as bundles on $F$.
\end{lemma}
Next it is time to find the set $\widetilde{W}_\lambda$ associated to our variational problem; in order to mimick the definition made in the case of first order field theory, we want to define this set according to the formula
\begin{equation}\label{Eq:DeffWL}
  \left.\widetilde{W}_\lambda\right|_{f}:=\left(\left.\lambda\right|_f+\left.\left(I\cap\wedge^mF\right)\right|_f\right)\cap\left.\left[\left(\wedge^m_pF\right)^V\right]\right|_f,\qquad\forall f\in F.
\end{equation}
We have some requirements to impose in order to ensure it exists and has nice properties, namely:
\begin{itemize}
\item There must exists a subbundle $I\subset\wedge^\bullet F$ such that $\cI$ is generated (in the sense of Definition \ref{Def:EDSGenerated}) by its sections.
\item $\widetilde{W}_\lambda$ must be a bundle on $F$, meaning in particular that the intersection
  \[
  I\cap\left(\wedge^m_pF\right)^V
  \]
  must have constant rank.
\item The set $\widetilde{W}_\lambda$ must contain all the relevant data belonging to the variational problem, i.e., both the Lagrangian form and every multiple of the algebraic generators of $\cI$.
\end{itemize}
\begin{example}[On the last condition]
  Let us consider the EDS
  \[
  \left(\mR^4\stackrel{\pi}{\longrightarrow}\mR^2\rightarrow\mR^2,0,\left\langle \theta\right\rangle _{\text{diff}}\right),
  \]
  where the global coordinates on $\mR^4$ are $\left(x,y,u,v\right)$, $\pi\left(x,y,u,v\right)=\left(x,y\right)$ is the projection map and moreover $\theta:=\dif u\wedge\dif v$. Then the underlying diagram becomes
  \[
  \begin{diagram}
    \node{\text{id}^*\left(\wedge^2_2\mR^4\right)}\arrow{s,l}{\bar\tau}\arrow{e,t}{}\node{\wedge^2_2\mR^4}\arrow{s,r}{\bar\tau}\\
    \node{\mR^4}\arrow{e,b}{}\node{\mR^4}\arrow{e,b}{\pi}\node{\mR^2}
  \end{diagram}
  \]
  and we have that
  \[
  \text{id}^*\left(\wedge^2_2\mR^4\right)=\left(\wedge^2_2\mR^4\right)^V=\wedge^2_2\mR^4;
  \]
  then
  \[
  \widetilde{W}=\mR\theta\cap\wedge^2_2\mR^4=0
  \]
  and the multisymplectic structure results trivial: Every $2$-vector is solution of the underlying equations. On the other side, the initial EDS represents the PDE
  \[
  u_xv_y-u_yv_x=0
  \]
  whose solutions are a proper subset of the solutions of the trivial multisymplectic structure. This must be solved by allowing the forms on $\mR^4$ to be of higher vertical degree. 
\end{example}
\begin{example}[On the last condition II]\label{Ex:LastCondII}
  The subtleties we could found related to the last item could also be illustrated by the following toy model
  \[
  \left(\mR^5\rightarrow\mR^3\rightarrow\mR^2,0,\left\langle \theta,\Gamma_1,\Gamma_2\right\rangle _{\text{diff}}\right),
  \]
  where the fibration structure is given by
  \[
  \left(x,y,u,p,q\right)\stackrel{\pi_{10}}{\longmapsto}\left(x,y,u\right)\stackrel{\pi}{\longmapsto}\left(x,y\right)
  \]
  and $\theta:=\dif u-p\dif x-q\dif y,\Gamma_1:=\dif p\wedge\dif x,\Gamma_2:=\dif q\wedge\dif y$. Then $\left(V\pi_{10}\right)^0=\left\langle \dif x,\dif y,\dif u\right\rangle _{\text{alg}}$ and so
  \[
  \left(\wedge^2_2\mR^5\right)^V=\left\{\left(x,y,u,p,q,a\dif x\wedge\dif y+b\dif u\wedge\dif x+c\dif u\wedge\dif y\right):\left(x,y,u,p,q,a,b,c\right)\in\mR^8\right\};
  \]
  it is evident that this set, although fulfilling the two first requirements above, does not meet the third, because it does not contain multiples of the generators $\Gamma_i,i=1,2$.
\end{example}
The last condition sets a constraint on the number $p$: It must be large enough to allow $\widetilde{W}_\lambda$ to fulfills it. So, in order to obtain a well-defined quantity, let us take $p$ as the minimum integer making it true. Additionally, it sets a constraint on the bundle $F$ as fibration on $E$, because the bundle $\pi_{10}^*\left(\wedge^m_pE\right)$ does not contain any form in the $\left(V\pi_{10}\right)^0$-direction; in fact, we could reformulate the bundles in Example \ref{Ex:LastCondII} in order to find a well-behaved variational problem: It is just enough to change the bundle $\mR^5\rightarrow\mR^3\rightarrow\mR^2$ by $\mR^5\rightarrow\mR^5\rightarrow\mR^2$.

\begin{defc}[Admissible variational problem]
  With the notation introduced above, we will say that the variational problem $\left(F\rightarrow E\rightarrow M,\lambda,\cI\right)$ is \emph{admissible} if and only if $\cI$ is linearly generated by a subbundle $I\subset\wedge^\bullet F$, the set $I\cap\left(\wedge^m_pF\right)^V$ has constant rank, and there exists an integer $p\leq m$ such that $\widetilde{W}_\lambda$ contains the Lagrangian form and every multiple of a set of algebraic generators of $\cI$ in $\left(\wedge^m_pF\right)^V$.
\end{defc}

The main purpose of the notion of admissibility for variational problems is to set the next result. It is formulated by using the notion of \emph{covariant Lepage-equivalent variational problem}, borrowed from \cite{GotayCartan} (see also \cite{Gotay1991203}.) In short, it means that every extremal for the variational problem $\left(W_\lambda,\Theta_\lambda,0\right)$ projects onto an extremal of the original variational problem via $p_1:W_\lambda\rightarrow F$.
\begin{proposition}
  Let $\left(F,\lambda,\cI\right)$ be an admissible variational problem. Then $\left({W}_\lambda,\Theta_{\lambda},0\right)$ is a covariant Lepage-equivalent variational problem.
\end{proposition}
\begin{proof}
  The proof goes in a similar way to the discussion of the canonical Lepage-equivalent problem in \cite{GotayCartan}. We need to prove that
  \begin{enumerate}[(i)]
  \item for every extremal $\gamma$ of $\left(W_\lambda,\Theta_\lambda,0\right)$ projecting onto an extremal of $\left(F,\lambda,\cI\right)$, we have that $\phi:=p_1\circ\gamma$ verifies $\phi^*\cI=0$, and\label{Item1}
  \item $\psi^*\lambda=\sigma^*\Omega_\lambda$ for every section $\sigma:M\rightarrow W_\lambda$ such that $\psi:=p_1\circ\sigma$ is an integral section of $\cI$.\label{Item2} 
  \end{enumerate}
  For \ref{Item1}, we will use that admissibility means that every generator (or multiple of it) $\beta$ of $\cI$ belongs to $W_\lambda$, in the sense that $\left.\beta\right|_f\in W_\lambda$ for every $f\in F$ where it is defined, and so the curve
  \[
  \beta^\alpha_t:=\alpha+t\left.\beta\right|_f,\qquad t\in\mR
  \]
  for every $\alpha\in\left.W_\lambda\right|_f$; therefore the vector field
  \[
  \delta\beta:\alpha\mapsto\left.\frac{\overrightarrow{\text{d}\beta^\alpha_t}}{\text{d}t}\right|_{t=0}
  \]
  is tangent to $W_\lambda$, and it meets the formula
  \[
  \delta\beta\lrcorner\Omega_\lambda=p_1^*\beta
  \]
  just as in the cotangent bundle case. Thus
  \begin{align*}
    \phi^*\beta&=\gamma^*p_1^*\beta\\
    &=\gamma^*\left(\delta\beta\lrcorner\Omega_\lambda\right)=0
  \end{align*}
  for every generator $\beta$ of $\cI$; it implies that $\phi^*\cI=0$. Then \ref{Item2} follows by taking into account the formula defining the canonical $m$-form, namely
  \[
  \sigma^*\left(\left.\Theta_\lambda\right|_\alpha\right)=\sigma^*p_1^*\alpha=\psi^*\alpha,
  \]
  and that $\alpha\in W_\lambda$ iff $\alpha=\lambda+\rho$ for some $\rho\in\cI$. Now with \ref{Item2} at our disposal, it is immediate to see that if $\phi_t$ is a curve in the integral sections of $\cI$ and $\gamma_t$ is any section of $W_\lambda$ covering it through $p_1$ such that $\gamma_0$ is an extremal of $\left(W_\lambda,\Theta_\lambda,0\right)$, we will have that $\phi_t^*\lambda=\omega_t\Theta_\lambda$ and so
  \[
  \left.\frac{{\text{d}}}{\text{d}t}\left(\int_M\phi_t^*\lambda\right)\right|_{t=0}=\left.\frac{{\text{d}}}{\text{d}t}\left(\int_M\gamma_t^*\Theta_\lambda\right)\right|_{t=0}=0
  \]
  because $\gamma_0$ was assumed to be extremum.
\end{proof}
As Gotay pointed out in the cited work, there is no warranty on the Lepage-equivalent problem to have the same solutions as the original, namely, there is no general proof of the so called \emph{contravariance} of this Lepage-equivalent problem. In the following we will take this as granted; in practice, this condition must be verified in each particular case separately, except in the case of classical variational problems, where proofs of the contravariance of this Lepage-equivalent can be found in the literature, see for instance \cite{GotayCartan,KrupkaVariational}. In the examples of Section \ref{Sec:Examples}, we provide a proof of this property in those cases where it was necessary. 

% Additionally, we will suppose that there exists a section $h:\pi_{10}^*\left(\wedge^m_pE/\wedge^m_{p-1}E\right)\rightarrow\pi_{10}^*\left(\wedge^m_pE\right)$ and a submanifold $M_\lambda\subset\pi_{10}^*\left(\wedge^m_pE\right)$ such that
% \[
% W_\lambda=\mathop{\text{Im}}{h}\cap M_\lambda.
% \]

\begin{defc}[Restricted Hamiltonian system associated to a variational problem]\label{Def:RestHamSsytForVarsProb}
  The \emph{restricted Hamiltonian system associated to the admissible variational problem }$\left(F,\lambda,\cI\right)$ is the (pre) multisymplectic manifold $\left(W_\lambda,\Theta_\lambda\right)$.
\end{defc}

Whenever Lepage-equivalent problem $\left(W_\lambda,\Theta_\lambda,0\right)$ is contravariant, the associated restricted Hamiltonian system has the same extremals of the original problem, and so its substitution keeps the solution set. It is a quite desirable property, because it is our interest to work with a problem without loosing crucial information.

\section{Tentative Gotay algorithm for general variational problems}\label{Sec:GotayAlgorithm}

\subsection{Introduction}
It is our purpose in the following paragraphs to describe an algorithm to deal with solutions of a restricted Hamiltonian system $\left(W,\Omega\right)$
\begin{equation}\label{Eq:RestHamTentative}
  Z\lrcorner\Omega=0,\qquad Z\lrcorner p^*\eta=1
\end{equation}
of similar nature to Gotay algorithm; here we are taking as granted that $W$ is a bundle $p:W\rightarrow M$ on spacetime, and $\eta$ is a volume form on $M$. Now the constraint algorithms of classical mechanics are designed for the search of the so called \emph{final constraint manifold}, which can be characterized by the property that through every of its points passes at least one solution of the underlying equations of motion. When dealing with (analytical) systems of PDE expressed as linear Pfaffians, a condition ensuring integrability is \emph{involutivity}, and there exists an algorithm designed to find another (involutive) linear Pfaffian (defined perhaps in a subset) whose solutions induces solutions of the original \cite{BCG,CartanBeginners,nkamran2000}. In particular, the subset obtained from this procedure must be included into the final constraint manifold, so it gives a kind of Gotay algorithm in this context.

There exists another reason to develop such an algorithm: In \cite{EcheverriaEnriquez:1998jr,DeLeon2005839} a procedure of this sort is given for deal with solutions of Equation \eqref{Eq:RestHamTentative}, which face two obstacles: The tangency condition (arising when asks the factors of $Z$ to be vectors tangent to the constraint manifold) and the integrability condition, the latter being a distinguished feature of this context, absent for dimensional reasons when working with these kind of algorithms in classical mechanics. The former condition is an expected outcome when dealing with ``contravariant'' elements, where no natural pullback exists; additionally, the integrability condition must be implemented by means of Lie brackets, a not very efficient method for work with these type of conditions\footnote{In order to have a look to the kind of difficulties people faces in dealing with integrability issues from Frobenius viewpoint, see for example \cite{1751-8121-46-10-105201} and references therein.}. The dual perspective offers advantages in both aspects: Forms have natural pullback, and the existence of exterior differentiation yields to methods for deal with integrability issues. 

Finally we would bring to the attention of the reader an additional advantage of the approach chosen here to (a procedure analogous to) Gotay algorithm: When applied to field theory, it does not depends on the choice of slices of the bundle of fields (the constant time leaves.) So we are working with a \emph{covariant constraint algorithm} for field theory. Additionally it could be considered as an approach (from the viewpoint of Cartan's EDS) to the Dirac theory of constraints, complementary to the analysis of the same subject carried out in \cite{Seiler95involutionand} from the viewpoint of Janet-Riquier theory of formal integrability.

\subsection{Formulation of the algorithm}
So, the idea is to use a known algorithm, dubbed \emph{Cartan algorithm} \cite{CartanBeginners,Hartley:1997:IAN:2274723.2275278}, useful when dealing with linear Pfaffian EDS, in order to get rid of the integrability condition; the flowchart shown in Figure \ref{Fig:CartanAlgorithm} sketch it, and further details can be found in Appendix \ref{App:EDSInvolution}.
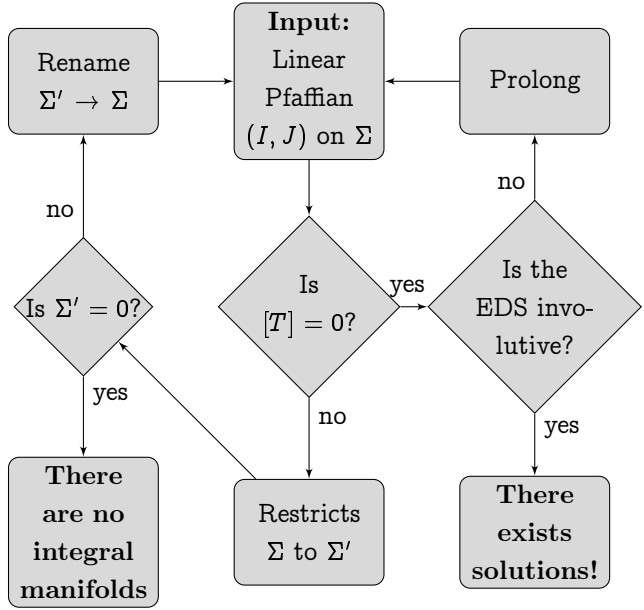
\begin{figure}[h]
  \begin{tikzpicture}[node distance = 3cm, auto]
    % Place nodes
    \node [block] (init) {\textbf{Input:}\\Linear Pfaffian $\left(I,J\right)$ on $\Sigma$};
    \node [block, left of=init] (rename) {Rename $\Sigma'\to\Sigma$};
    \node [block, right of=init] (prolong) {Prolong};
    \node [decision, below of=init] (zerotorsion) {Is $\left[T\right]=0$?};
    \node [decision, left of=zerotorsion] (empty) {Is $\Sigma'=0$?};
    \node [decision, right of=zerotorsion] (involutive) {Is the EDS involutive?};
    \node [block, below of=zerotorsion] (restrict) {Restricts $\Sigma$ to $\Sigma'$};
    \node [block, below of=involutive] (existence) {\textbf{There exists solutions!}};
    \node [block, below of=empty] (failure) {\textbf{There are no integral manifolds}};
    % \node [block, left of=evaluate, node distance=3cm] (update) {update model};
    % \node [decision, below of=evaluate] (decide) {is best candidate better?};
    % \node [block, below of=decide, node distance=3cm] (stop) {stop};
    % Draw edges
    \path [line] (init) -- (zerotorsion);
    \path [line] (zerotorsion) -- node [near start] {yes} (involutive);
    \path [line] (zerotorsion) -- node [near start] {no} (restrict);
    \path [line] (involutive) -- node [near start] {yes} (existence);
    \path [line] (involutive) -- node [near start] {no} (prolong);
    \path [line] (empty) -- node [near start] {yes} (failure);
    \path [line] (empty) -- node [near start] {no} (rename);
    \path [line] (restrict) -- (empty);
    \path [line] (rename) -- (init);
    \path [line] (prolong) -- (init);
    % \path [line,dashed] (system) |- (evaluate);
  \end{tikzpicture}
  \caption{The flowchart for Cartan algorithm (from \cite{CartanBeginners}.)}\label{Fig:CartanAlgorithm}
\end{figure}
The way to fit in this scheme is to introduce the canonical contact structure on the Grassmann bundle $\mathsf{Gr}_m\left(W,p^*\eta\right)$; the canonical structure is differentially generated by the sets of forms
\[
\left.I\right|_{\left(w,E\right)}:=\pi^*\left(E^0\right)\subset\Omega^1\left(G_m\left(TW,p^*\eta\right)\right).
\]
The crucial fact is that equations \eqref{Eq:RestHamTentative} could be used to define a subbundle $G_0$ of $G_m\left(TW,p^*\eta\right)$, as the zero set of the mapping
\[
E\mapsto X\lrcorner\Omega\qquad\text{if and only if}\qquad X\in\mathfrak{X}^nW\text{ such that }\left[X\right]=E,
\]
where $\left[\cdot\right]$ indicates the subspace spanned by the components of the decomposable $m$-multivector in it. In the case of restricted Hamiltonian system associated to a classical field theory, the subset $G_0$ is composed by the $n$-planes defined by the PDE system of Hamilton equations, so we could introduce the following notation.

\begin{defc}
  The set $G_0$ will be called the \emph{Hamilton submanifold} of $G_m\left(TW,p^*\eta\right)$; the EDS $\cI_0$ induced on $G_0$ by the canonical structure will be referred as \emph{Hamilton (linear) Pfaffian}.
\end{defc}

\begin{note}[Local calculations]
  It is convenient to introduce a local description of these objects in order to have an intuition on these matters. Recall that every adapted coordinate chart $\left(x^i,u^A\right)$ on $W$ induces the coordinates $\left(x^i,u^A,p_j^B\right)$ on an open set $U\subset G_n\left(TW,p^*\eta\right)$ such that every $E\in U$ fulfills
  \[
  E=\left\langle \frac{\partial}{\partial x^i}+p^A_i\left(E\right)\frac{\partial}{\partial u^A}:i=1,\cdots,n\right\rangle .
  \]
  Then if we parametrize the set of solutions of Equation \eqref{Eq:RestHamTentative} by $Z:=Z_1\wedge\cdots\wedge Z_n$, where
  \[
  Z_i:=\frac{\partial}{\partial x^i}+Z^A_i\frac{\partial}{\partial u^A},
  \]
  then $G_0$ is described by the equations
  \[
  p^A_i=Z^A_i, 
  \]
  and the generators of $\cI$ reads
  \[
  \theta^A:=\dif u^A-Z^A_k\dif x^k.
  \]
\end{note}
\begin{note}
  From now on we will suppose that $\Omega$ in our restricted Hamiltonian systems meets a \emph{regularity criteria}, namely, that $G_0$ is a subbundle of $\mathsf{Gr}_m\left(W\right)$.
\end{note}
The relevance of this concept lies on the following result; in part it justifies the choice of language made in the previous definition.
\begin{prop}\label{Prop:IntSectionsI0}
  The integral sections for the bundle $G_0\rightarrow M$ of the EDS with independence condition
  \[
  \left(\cI_0,p^*\eta\right)
  \]
  are in one-to-one correspondence with the solutions of the Hamilton equations \eqref{Eq:RestHamTentative}.
\end{prop}
This proposition finish the two-way road going from a variational problem to a linear Pfaffian EDS, schematically presented in Figure \ref{Fig:DifferentSolutions}.
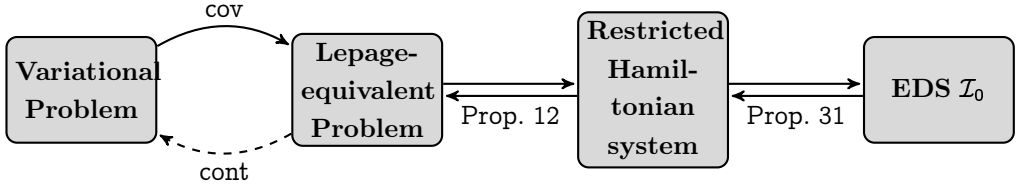
\begin{figure}[h]
  \begin{tikzpicture}[->,>=stealth',shorten >=1pt,auto,node distance=3.8cm,thick]
    % Place nodes
    \node [block] (vp) {\textbf{Variational\\Problem}};
    \node [block, right of=vp] (lep) {\textbf{Lepage-equivalent Problem}};
    \node [block, right of=lep] (rhs) {\textbf{Restricted Hamiltonian system}};
    \node [block, right of=rhs] (eds) {\textbf{EDS }$\cI_0$};
    % Draw edges
    \path%[every node/.style={font=\sffamily\small}]
    (lep) edge [bend left,dashed] node[below] {cont} (vp)
    (vp) edge [bend left] node[above] {cov} (lep);
    \draw[transform canvas={yshift=0.5ex},->] (lep) -- (rhs);
    \draw[transform canvas={yshift=-0.5ex},->] (rhs) -- (lep) node[below,midway] {Prop.\ \ref{Prop:OursEquivalentOnVars}};
    \draw[transform canvas={yshift=0.5ex},->] (rhs) -- (eds);
    \draw[transform canvas={yshift=-0.5ex},->] (eds) -- (rhs) node[below,midway] {Prop.\ \ref{Prop:IntSectionsI0}};
  \end{tikzpicture}
  \caption{The different representations of the dynamical problem}\label{Fig:DifferentSolutions}
\end{figure}
It means that we can translate extremals of the original variational problem into extremals of the Lepage-equivalent problem, they can be rewritten as solutions of a restricted Hamiltonian system, and finally into integral sections of a linear Pfaffian EDS; the maps inducing these correspondence are fully understood, and were described early. Therefore, under the assumption of contravariance of the chosen Lepage-equivalent problem, it is the same to work with the original variational problem, its Lepage-equivalent, the restricted Hamiltonian system or with the associated Hamilton Pfaffian. By choosing the final option, we are in position to use the Cartan algorithm: After the first sequence of absortion of torsion, elimination of $0$-forms and a unique prolongation, we will obtain the following diagram
\begin{figure}[h]
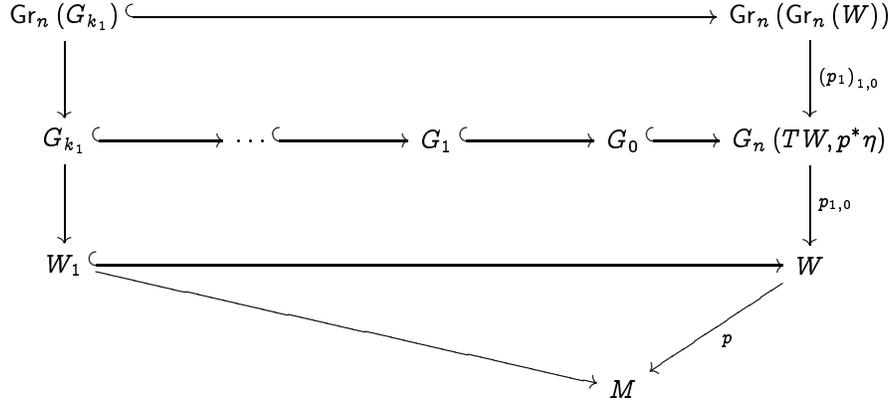

  \[
  \begin{diagram}
    \node{\mathsf{Gr}_n\left(G_{k_1}\right)}\arrow[4]{e,t,J}{}\arrow{s,l}{}\node[4]{\mathsf{Gr}_n\left(\mathsf{Gr}_n\left(W\right)\right)}\arrow{s,r}{\left(p_{1}\right)_{1,0}}\\
    \node{G_{k_1}}\arrow{e,b,J}{}\arrow{s,l}{}\node{\cdots}\arrow{e,b,J}{}\node{G_1}\arrow{e,b,J}{}\node{G_0}\arrow{e,b,J}{}\node{G_n\left(TW,p^*\eta\right)}\arrow{s,r}{p_{1,0}}\\
    \node{W_1}\arrow[4]{e,t,J}{}\arrow{seee,b}{}\node[4]{W}\arrow{sw,r}{p}\\
    \node[4]{M}
  \end{diagram}
  \]
  \caption{A sequence of restrictions plus a prolongation}
  \label{Fig:SeqRestrictionProl}
\end{figure}
where we use the shorthand $\mathsf{Gr}_n\left(F\right):=G_n\left(TW,q^*\eta\right)$ for any bundle $q:F\rightarrow M$, and the convention that when using the Grassmannian of such a bundle, the induced morphisms reads
\[
\begin{diagram}
  \node{\mathsf{Gr}_n\left(F\right)}\arrow{e,t}{q_{1,0}}\arrow{se,b}{q_1}\node{F}\arrow{s,r}{q}\\
  \node[2]{M}
\end{diagram}
\]
Under smoothness assumptions, a successful termination of the algorithm yields to the diagram
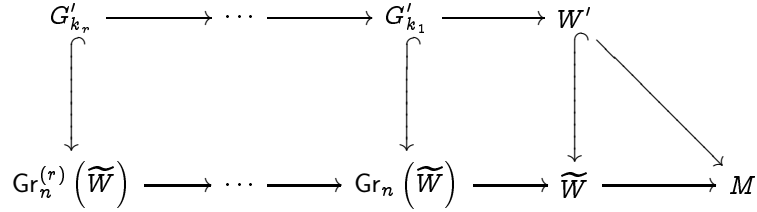
\[
\begin{diagram}
  \node{{G_{k_r}'}}\arrow{s,l,J}{}\arrow{e,t}{}\node{\cdots}\arrow{e,t}{}\node{{G_{k_1}'}}\arrow{e,t}{}\arrow{s,l,J}{}\node{{W'}}\arrow{se,t}{}\arrow{s,l,J}{}\\
  \node{\mathsf{Gr}_n^{\left(r\right)}\left(\widetilde{W}\right)}\arrow{e,b}{}\node{\cdots}\arrow{e,b}{}\node{\mathsf{Gr}_n\left(\widetilde{W}\right)}\arrow{e,b}{}\node{\widetilde{W}}\arrow{e,b}{}\node{M}
\end{diagram}
\]
where the primes are indicating that every prolongation could induce further restrictions on the spaces at right, and
\[
\mathsf{Gr}^{\left(n\right)}\left(F\right)=\mathsf{Gr}\left(\mathsf{Gr}^{\left(n-1\right)}\left(F\right)\right),\quad\mathsf{Gr}^{\left(0\right)}\left(F\right)=F.
\]
After imposing the differential conditions induced by the multiple contact structures underlying the spaces in this diagram, we could extract the following information:
\begin{itemize}
\item ${W}'$ provide us the restrictions we need to impose on the dependent variables.
\item The set ${G_{k_1}'}$ restricts the components of the multivector solution $Z$.
\item The sets ${G_{k_j}'},j\geq2$ give us restrictions on the derivatives of the components of $Z$.
\end{itemize}

\begin{note}[A discussion concerning contravariance of Lepage-equivalent problem]
  We have two remarks to make at this points, both concerning the proof of contravariance:
  \begin{itemize}
  \item When an existence result is at our disposal, there exists some arguments at hands in order to ensures contravariance of Lepage-equivalent problem $\left(W_\lambda,\Theta_\lambda,0\right)$. The idea is to see the behaviour of regular integral elements of the final space $G_{k_r}\subset \mathsf{Gr}_m^{\left(r\right)}\left(\widetilde{W}\right)$ under the differential of the map $\mathsf{Gr}_m^{\left(r\right)}\rightarrow W'\hookrightarrow\widetilde{W}\rightarrow F$; because an integral section of the original variational problem $\left(F,\lambda,\cI\right)$ gives rise to $m$-planes on $F$, the question reduces to see if these $m$-planes are covered by regular integral elements of the Lepage-equivalent problem via the above mentioned map. Because of the covariance, every projection of a regular element must be a infinitesimal solution of the Euler-Lagrange equations; nevertheless, it is not always clear whether every of the planes tangent to solutions of Euler-Lagrange equations has a regular integral element on it.
  \item It is important to point out that in general the proof of contravariance (see Subsections \ref{Sec:ContravarianceInt1} and \ref{Sec:ContravarianceInt2}) yields to the study of a new EDS, namely, the pullback of the Hamilton-Cartan EDS to the subbundle $p_1^{-1}\left(\text{Im}\,s\right)$ for $s:M\rightarrow F$ a solution of the Euler-Lagrange equations. A well-chosen Lepage-equivalent problem gives rise to an EDS which can be easily solved, in order to establish the contravariance; for example, the classical variational problem for field theory yields to a case like this, and the multitude of Cartan forms that can be found for it could be associated to the differents solutions that can be discovered to this underlying EDS.
  \end{itemize}
\end{note}

\section{Examples}
\label{Sec:Examples}

We are ready to see how the scheme discussed in the previous sections works on particular examples. They were chosen in order to try to make connections with systems studied in the literature from another viewpoint.

\subsection{Examples from classical mechanics}\label{Sec:ClassicalExample}

Although designed in order to get rid of constraints arising from field theory, it is interesting to note that this algorithm can be used in the classical mechanics realm.

\subsubsection{Geometrical setting}

The underlying geometrical setting is best summarized by the following diagram
\[
\begin{diagram}
  \node{\left(\text{id}\times\tau_Q\right)^*\left(T^*\left(\mR\times Q\right)\right)}\arrow{e,t}{p_2}\arrow{s,l}{p_1}\node{\wedge^1_1\left(\mR\times Q\right)}\arrow{s,r}{\bar{\tau}_{\mR\times Q}}\\
  \node{\mR\times TQ}\arrow{e,b}{\text{id}\times\tau_Q}\node{\mR\times Q}\arrow{e,b}{}\node{\mR}
\end{diagram}
\]
and the additional identifications
\begin{align*}
  \left(\text{id}\times\tau_Q\right)^*\left(T^*\left(\mR\times Q\right)\right)&=\mR\times TQ\oplus_{\mR\times Q}T^*\left(\mR\times Q\right),\\
  \wedge^1_1\left(\mR\times Q\right)&=T^*\left(\mR\times Q\right).
\end{align*}
Thus we will work on the subbundle of the space of forms
\[
\left(\text{id}\times\tau_Q\right)^*\left(T^*\left(\mR\times Q\right)\right)=\left\{\left(t,q^i,v^j,\alpha\dif t+p_i\dif q^i\right)\right\}\subset\wedge^1_1\left(\mR\times TQ\right)
\]
with the inclusion given by the identification
\[
\left(\text{id}\times\tau_Q\right)^*\left(T^*\left(\mR\times Q\right)\right)=\left(\wedge^1_1\left(\mR\times TQ\right)\right)^V;
\]
therefore we have a multisymplectic structure induced by the canonical structure
\[
\left.\Omega\right|_{\left(t,q,v;\alpha,p,r\right)}:=\dif \alpha\wedge\dif t+\dif p_i\wedge\dif q^i+\dif r_i\wedge\dif v^i\in\Omega^2\left(\wedge^1_1\left(\mR\times TQ\right)\right).
\]
Let us suppose that we have a Lagrangian system described by the Lagrangian function $L\in C^\infty\left(\mR\times TQ\right)$. Our immediate task is to define the set $\widetilde{W}_{L}$ according to the formula \eqref{Eq:DeffWL}; because the contact structure is generated by the sections of the subbundle
\[
\mR\left(\dif q^i-v^i\dif t\right)\subset\wedge^1_1\left(\mR\times TQ\right),
\]
we will have that
\[
\widetilde{W}_{L}=\left\{L\dif t+p_i\left(\dif q^i-v^i\dif t\right)\right\}\subset\left(\text{id}\times\tau_Q\right)^*\left(T^*\left(\mR\times Q\right)\right),
\]
and the (pre)symplectic $2$-form will be given by $\Omega^L:=\left.\Omega\right|\widetilde{W}_L$, namely
\[
\left.\Omega^L\right|_{\left(t,q,v,p\right)}:=\dif\left(L-p_iv^i\right)\wedge\dif t+\dif p_i\wedge\dif q^i\in\Omega^2\left(\widetilde{W}_L\right).
\]
The restricted Hamiltonian system has the Hamilton equations
\begin{equation}\label{Eq:EcsHam1D}
  Z\lrcorner\Omega^L=0,\qquad Z\lrcorner\dif t=1
\end{equation}
which define the vector field $Z$. In the following paragraphs we will use this approach to work with some singular Lagrangian systems.

\subsubsection{Example from \cite{Sundermeyer}}

Let us consider the singular Lagrangian
\[
L\left(q^1,q^2,v^1,v^2\right):=\frac{1}{2}\left(v^1\right)^2+q^2v^1+\left(1-\alpha\right)q^1v^2+\beta\left(q^1-q^2\right)^2
\]
where $\alpha$ and $\beta$ are some constants. The coordinates on $\mathsf{Gr}_1\left(\widetilde{W}_L\right)$ are
\[
\left(t,q^i,v^i,p_i;Z^{q^i},Z^{v^i},Z^{p_i}\right)
\]
if and only if they represent the line spanned by the vector
\[
Z:=\frac{\partial}{\partial t}+Z^{q^i}\frac{\partial }{\partial q^i}+Z^{v^i}\frac{\partial }{\partial v^i}+Z^{p_i}\frac{\partial }{\partial p_i}\in T_{\left(t,q^i,v^i,p_i\right)}\widetilde{W}_L.
\]
The Hamilton equations \eqref{Eq:EcsHam1D} yield to
\begin{align*}
  Z^{p_1}&=\beta\left(q^1-q^2\right)+\left(1-a\right)v^2,\\
  Z^{p_2}&=v^1-\beta\left(q^1-q^2\right),\\
  Z^{q^1}&=v^1,\\
  Z^{q^2}&=v^2
\end{align*}
alongside the restrictions
\[
p_1=q^2+v^2,\quad p_2=\left(1-\alpha\right)q^1.
\]
These equations define a subbundle $G_0$ of the Grassmann bundle $\mathsf{Gr}_1\left(\widetilde{W}_L\right)$ projecting onto a submanifold $W_0\subset\widetilde{W}_L$, which is defined by the two last equations, fitting in the commutative diagram
\[
\begin{diagram}
  \node{G_0}\arrow{s,l}{}\arrow{e,t,J}{}\node{\mathsf{Gr}_1\left(\widetilde{W}_L\right)}\arrow{s,r}{}\\
  \node{W_0}\arrow{e,b,J}{}\node{\tilde{W}_L}
\end{diagram}
\]
The restriction of the contact structure, locally generated by the $1$-forms
\[
\theta^{q^i}:=\dif q^i-Z^{q^i}\dif t,\quad\theta^{v^i}:=\dif v^i-Z^{v^i}\dif t,\quad\theta^{p_i}:=\dif p_i-Z^{p_i}\dif t
\]
gives rise to the initial EDS $\cI_0$ on $G_0$; in particular, it contains $0$-forms yielding to the restrictions
\[
Z^{v^1}=\beta\left(q^1-q^2\right)-\alpha v^2,\quad\beta\left(q^1-q^2\right)-\alpha v^1=0.
\]
In the reference cited above, the analysis turns out to depends onto the differents values of these numbers; it is true in our case, and we will proceed accordingly.
\begin{itemize}
\item\textbf{Case $\left(1\right)$: $a\not=0$.} In this case we can restrict ourselves to a subbundle $G_1\subset G_0$ fibred on the submanifold $W_1\subset W_0$ defined there by the second equation above. Under the assumption $\alpha\not=0$, on $G_1$ the velocity $v^1$ can be expressed as function of the rest of the coordinates. The EDS $\cI_1$ determined by pullback of $\cI_0$ contains another $0$-form, which gives us another constraint, namely
  \[
  \left(\alpha^2-\beta\right)\left(av^2-\beta\left(q^1-q^2\right)\right)=0.
  \]
  We have here a couple of subcases we need to take care of:
  \begin{itemize}
  \item \textbf{Case $\left(1.A\right)$: $\beta\not=\alpha^2$.} Then it shows up an additional subbundle $G_2$ fibred on the submanifold $W_2\subset W_1$ determined by the formula
    \[
    v^2=\frac{\beta\left(q^1-q^2\right)}{\alpha};
    \]
    the EDS $\cI_2$ induced here is involutive, yielding to the equations of motion
    \[
    \dot{q}^1=\dot{q}^2=\frac{\beta}{\alpha}\left(q^1-q^2\right).
    \]
  \item \textbf{Case $\left(1.B\right)$: $\beta=\alpha^2$.} In this case the EDS $\cI_1$ on $G_1$ is involutive, and this means that $W_1$ is the final constraint submanifold; the equations of motion there become
    \[
    \dot{q}^1=\alpha\left(q^1-q^2\right)
    \]
    with no further restrictions on the coordinate $q^2$.
  \end{itemize}
\item \textbf{Case $\left(2\right)$: $a=0$.} The EDS $\cI_0$ contains a set of $0$-forms yielding to the equations
  \[
  \beta\left(q^2-q^1\right)=0,\qquad Z^{v^1}=0,
  \]
  and the first of these gives rise to additional constraints under further assumptions:
  \begin{itemize}
  \item \textbf{Case $\left(2.A\right)$: $\beta\not=0$.} In such case we have that the subbundle $G_0$, fibered on the submanifold $W_0$, is determined by the equation
    \[
    q^1=q^2.
    \]
    The pullback along these restrictions gives rise to the new constraint
    \[
    v^1-v^2=0,
    \]
    which determines another subbundle $G_1'\subset G_0$ and the corresponding base submanifold $W_1'\subset W_0$; on $G_1'$ there exists an induced EDS $\cI_1'$. A new restriction of these geometrical structures gives rise to the constraint
    \[
    Z^{v^2}=0,
    \]
    thus defining a subbundle $G_2'\subset G_1'$ which is fibered on the same submanifold $W_2'=W_1'$. The EDS $\cI_2'$ is involutive, and the corresponding equations of motion are
    \[
    \ddot{q}_2=0.
    \]
  \item \textbf{Case $\left(2.B\right)$: $\beta=0$.} The subbundle $G_0$ is given by
    \[
    Z^{v^1}=0
    \]
    with no restriction in the base, so $G_0$ is fibered on $\mR\times TQ$; the induced EDS $\cI_0$ is involutive, and yields to the sole equation of motion
    \[
    \ddot{q}^1=0.
    \]
  \end{itemize}
\end{itemize}

\subsection{Example: Maxwell equations}

We want to deal here with a physical example, not a toy model, in order to show the handiness of the algorithm; additionally, we will describe this theory using a non classical variational problem, trying to advertise on the advantages of the variational problems in Griffiths formulation.
\subsubsection{Geometrical setting}
Let $M$ be a pseudoriemannian manifold with metric $g$; the canonical volume $\eta$ has the simple form
\[
\eta=\sqrt{\abs{g}}\dif x^1\wedge\cdots\wedge\dif x^m
\]
in some coordinate neighborhood. Let us begin with a non standard example simple enough to carry out the previous manipulations to a successful end. It is \emph{non standard} in the sense that the underlying variational principle
\[
\left(\wedge{}^2M\oplus T^*M\rightarrow T^*M\rightarrow M,g^{pa}g^{qb}F_{ab}F_{pq}\dif^mx,F_{pq}\dif x^p\wedge\dif x^q-\dif A_p\wedge\dif x^p\right)
\]
is not a classical one, namely, coming from a variational problem on a jet space. In order to apply the previous scheme, it is necessary to perform the following identifications
\begin{align*}
  E&\leadsto T^*M\\
  J^1\pi&\leadsto\wedge{}^2M\oplus T^*M\\
  \pi_{10}&\leadsto P_2\text{ such that }P_2\left(F,A\right)=A;
\end{align*}
therefore the space where the restricted Hamiltonian system will live is given by the pullback bundle $P_2^*\left(\wedge{}^m_2\left(T^*M\right)\right)$, i.e.
\[
\begin{diagram}
  \node{P_{2}^*\left(\wedge{}^m_2\left(T^*M\right)\right)}\arrow{e,t}{p_2}\arrow{s,l}{p_1}\node{\wedge{}^m_2\left(T^*M\right)}\arrow{s,r}{\bar\tau^m_{T^*M}}\\
  \node{\wedge{}^2M\oplus T^*M}\arrow{e,b}{P_2}\node{T^*M}  
\end{diagram}
\]
The classical Lepage equivalent must be formulated on
\[
\left(\wedge^m_2\left(\wedge^2M\oplus T^*M\right)\right)^V:=\wedge^m_2\left(\wedge^2M\oplus T^*M\right)\cap\left(VP_2\right)^0
\]
which is the corresponding formula in this context to Equation \eqref{Eq:ClassicalVerticalSubBundle}. In terms of the adapted coordinates $\left(x^k,A_i,F_{kl}\right)$ on $\wedge^2M\oplus T^*M$, we will have that $\alpha\in\left(\wedge^m_2\left(\wedge^2M\oplus T^*M\right)\right)^V$ if and only if
\[
\alpha=p\eta+q^{ij}\dif A_i\wedge\eta_j
\]
for some $p,q^{ij}$; here we introduce, as before, the handy notation $\eta_{i}:=\partial_i\lrcorner\eta$. It defines the coordinates
\[
\alpha\mapsto\left(x^k,A_i,F_{ij},p,q^{ij}\right)
\]
on $\left(\wedge^m_2\left(\wedge^2M\oplus T^*M\right)\right)^V$, whose transformation properties are detailed by the next proposition.
\begin{proposition}
  Every coordinate change in $M$
  \[
  x^i\mapsto y^k=y^k\left(x\right)
  \]
  induces the coordinate change on $\left(\wedge^m_2\left(\wedge^2M\oplus T^*M\right)\right)^V$,
  \[
  \left(x^i,A_i,F_{ij},p,q^{ij}\right)\mapsto\left(y^k,B_k,G_{kl},r,s^{kl}\right),
  \]
  where
  \begin{align*}
    B_k&=\frac{\partial x^i}{\partial y^k}A_i\\
    G_{kl}&=\frac{\partial x^i}{\partial y^k}\frac{\partial x^j}{\partial y^l}F_{ij}\\
    p&=r+A_is^{jk}\frac{\partial^2 x^i}{\partial y^j\partial y^k}\\
    q^{ij}&=s^{kl}\frac{\partial x^i}{\partial y^k}\frac{\partial x^j}{\partial y^l}.
  \end{align*}
\end{proposition}

So we have the following result.
\begin{proposition}
  There exists an isomorphism
  \[
  \left(\wedge^m_2\left(\wedge^2M\oplus T^*M\right)\right)^V\simeq P_2^*\left(\wedge{}^m_2\left(T^*M\right)\right)
  \]
  as bundles on $\wedge^2M\oplus T^*M$.
\end{proposition}
\begin{proof}
  In the adapted coordinates introduced above the isomorphism reads
  \[
  p\eta+q^{ij}\dif A_i\wedge\eta_j\in\left.\wedge^m_2\right|_{\left(F,A\right)}\left(\wedge^2M\oplus T^*M\right)\mapsto\left(F_{ij}\dif x^i\wedge\dif x^j\oplus A_i\dif x^i,p\eta+q^{ij}\dif A_i\wedge\eta_j\right).\qedhere
  \]
\end{proof}
Up to now we proceed with no extra assumptions about $M$; from now on, in sake of simplicity, we will consider that we are working on $M=\mR^m$ with $g$ a constant metric.

\subsubsection{Contravariance for Lepage-equivalent problem of Maxwell equations}

Let us suppose that $\left(x^i\right)\mapsto\left(A_l,F_{ij}\right)$ is a solution of the Euler-Lagrange equations for Maxwell equations, namely
\[
F=\dif A,\qquad \dif\mathop{\ast}{F}=0;
\]
if these functions are replaced in the Hamilton-Cartan EDS, is obtained an EDS generated by the forms
\[
\alpha^l:=\dif P^{rl}\wedge\eta_r,\qquad\beta^{ij}:=\left(2F^{ij}-\left(-1\right)^{m+1}P^{ij}\right)\eta.
\]
The unique solution for this EDS under the independence condition $\eta\not=0$ is given by
\[
P^{ij}:=\frac{\left(-1\right)^{m+1}}{2}F^{ij},
\]
with no further restrictions arising from the annihilation of the forms $\alpha^l$, because of the Euler-Lagrange equations; namely, the lift of solutions from Euler-Lagrange equations to Hamilton-Cartan equations is uniquely determined, and thus it is a contravariant problem.

\subsubsection{Restricted Hamiltonian system for Maxwell equations}

We are now ready to define the bundle $\widetilde{W}_\lambda\rightarrow\wedge^2M\oplus T^*M$ via an analogous formula
\[
\left.\widetilde{W}_\lambda\right|_{\left(F,A\right)}:=\left.\lambda\right|_{\left(F,A\right)}+\left.I_{\text{M}}\right|_{\left(F,A\right)}\cap\left(\wedge^m_2\left(\wedge^2M\oplus T^*M\right)\right)^V,
\]
where $I_{\text{M}}$ is the subbundle of $\wedge^\bullet\left(\wedge^2M\oplus T^*M\right)$ whose sections generate the Maxwell EDS
\[
\cI_{\text{M}}:=\left\langle \Theta_1-\dif\Theta_2\right\rangle _{\text{diff}},
\]
where $\Theta_i,i=1,2$ is the canonical form on the first and second summand respectively. Then we have the following result.
\begin{proposition}\label{Prop:WinMaxwellEqs}
  The set $\widetilde{W}_\lambda\subset\left(\wedge^m_2\left(\wedge^2M\oplus T^*M\right)\right)^V$ is given locally by the equations $p=\left(g^{ap}g^{bq}F_{pq}-P^{pq}\right)F_{pq}$ and $q^{ij}=-2P^{ij}$ for some collection $\left(P_{ab}\right)$ of real numbers such that $P_{ij}+P_{ji}=0$. Equivalently we can say that $\widetilde{W}_\lambda$ is composed of the forms
  \[
  \alpha=\left[F^{pq}-\frac{\left(-1\right)^{m+1}}{2}q^{pq}\right]F_{pq}\eta+q^{ij}\dif A_i\wedge\eta_j
  \]
  such that $q^{ij}+q^{ji}=0$.
\end{proposition}
\begin{proof}
  In the working coordinates
  \[
  \left.I_{\text{M}}\right|_{\left(F,A\right)}=\left\{\beta\wedge\left(F_{pq}\dif x^p\wedge\dif x^q-\dif A_l\wedge\dif x^l\right):\beta\in\wedge^\bullet\left(\wedge^2M\oplus T^*M\right)\right\},
  \]
  meaning that
  \[
  \left.\widetilde{W}_\lambda\right|_{\left(F,A\right)}=\left\{\left(F^{pq}F_{pq}-P^{ab}F_{ab}\right)\eta+P^{pq}\eta_{pq}\wedge\dif A_l\wedge\dif x^l\right\}
  \]
  where it was introduced the usual terminology $F^{ab}:=g^{ap}g^{bq}F_{pq}$. By using the identity
  \[
  \dif x^p\wedge\eta_{ij}=\delta^p_j\eta_i-\delta^p_i\eta_j
  \]
  it could be written as
  \[
  \left.\widetilde{W}_\lambda\right|_{\left(F,A\right)}=\left\{\left(F^{ab}-P^{ab}\right)F_{ab}\eta+2P^{pq}\eta_{p}\wedge\dif A_q:P_{ab}+P_{ba}=0\right\};
  \]
  therefore $p=\left(F^{ab}-P^{ab}\right)F_{ab}$ and $q^{ij}=\left(-1\right)^{m+1}2P^{ij}$ are the equations describing this subset.
\end{proof}
Thus we have just obtained the first remarkable difference between this example and the classical case of first order field theory: The codimension of $\widetilde{W}_{L\eta}$ in $\left(\wedge^m_2J^1\pi\right)^V$ is $1$ in the latter, and the codimension of $\widetilde{W}_\lambda$ in $\left(\wedge^m_2\left(\wedge^2M\oplus T^*M\right)\right)^V$ is equal to $\frac{m\left(m+1\right)}{2}+1$ in the former case, because of the set of restrictions $q^{ij}+q^{ji}=0$. 
\begin{note}[On first order formalism]
  There exists a formulation of Maxwell equations where the fields $A$ and $F$ are considered as independent each other, the so called \emph{first order formalism} \cite{Sundermeyer}. It is immediate to verify that the set of restrictions \eqref{Eq:PFuncF}, arising from the zero-forms set of the Hamilton equations, defines a subbundle $\Sigma\subset\widetilde{W}_\lambda$ such that the Hamilton equations restricted to this set are equivalent to the equations of motion of the first order formalism.
\end{note}

\subsubsection{Implementation in Reduce via EDS package}

Let us describe the implementation of the previous algorithm in the computer algebra software called \emph{Reduce} \cite{Hearn:1967:RUM}, by using \emph{EDS} package. The first thing to do in this environment is to define the objects we work with, by using the command \verb|pform| of the \emph{EXCALC} package
\begin{reduce}
  pform {x(i), g(i,j), f(-i,-j), a(-i), p(i,j), u(i,j,-k), gdet, 
    xig(-i,-j,-k), xif(-i,-j,-k), xip(-i,j,k), xia(-i,-j), 
    FEqs(i,j), AEqs(i)}=0, 
  {th(i,j), restrictedham, ka(-i), kp(j,k), kf(-i,-j)}=1, 
  {s(-i,-j)}=2, {r(-i)}=3, {gamma,eta}=4, {theta}=5;
  tvector xi(i), va(i), vf(i,j);
\end{reduce}
where the indices $i,j,k,l$ runs from $1$ to $4$; so we will use the symbols
\StopVerbatimLineNos
\begin{verbatim}
x(i), f(-i,-j), a(-i), p(i,j)
\end{verbatim}
as variables on the manifold $\widetilde{W}_\lambda$. Additionally, the metric tensor $g^{ij}$ is set to be a diagonal Lorentz metric with signature $\left(-1,1,1,1\right)$. The (pre)multisymplectic form on this space will reads
\ResumeVerbatimLineNos
\begin{reduce}
  eta:=sqrt(abs(1/gdet))*d x(1)^d x(2)^d x(3)^d x(4);
  r(-i):=(@ x(i))_|eta;
  s(-i,-j):=(@ x(i))_|((@ x(j))_|eta);
  th(i,j):=d g(i,j)-u(i,j,-k)*d x(k);
  gamma:=(1/4)*g(i,j)*g(k,l)*f(-i,-k)*f(-j,-l)*eta+
  p(i,j)*s(-i,-j)^(f(-i,-j)*d x(i)^d x(j)-d a(-i)^d x(i));
  theta:=d gamma;
\end{reduce}
The components of $Z$ are defined through
\begin{reduce}
  va(j):=(@ a(-j));
  vf(i,j):=(@ f(-i,-j));
  xi(-i):=(@ x(i))+xip(-i,j,k)*(@ p(j,k))+
      xif(-i,-j,-k)*vf(j,k)+xia(-i,-j)*va(j);
\end{reduce}
and the equations of motion will be given by
\begin{reduce}
  restrictedham:=xi(-4)_|((xi(-3))_|((xi(-2))_|((xi(-1))_|theta)));
\end{reduce}
We try to choose these components in order to annihilate this $1$-form; it results that some additional requirements must be imposed by ensuring the existence of a solution, briefly, it is equivalent to 
\begin{equation}\label{Eq:PFuncF}
  P^{ij}=\frac{1}{32}g^{ik}g^{jl}F_{kl}.
\end{equation}
The restrictions on the components of the multivector $Z$ will result the following consistency conditions
\begin{subequations}
  \label{Eq:AllG}
  \begin{align}
    &Z^{P^{ij}}_k+Z^{P^{ji}}_k=0\label{Eq:G01}\\
    &Z^{F_{ij}}_k+Z^{F_{ji}}_k=0\label{Eq:G02},
  \end{align}
\end{subequations}
the relationship between $A$ and $F$,
\begin{align}\label{Eq:G03}
  Z^{A_i}_j-Z^{A_j}_i=8F_{ij}
\end{align}
and the true first pair of equations of motion, namely
\begin{align}\label{Eq:FirstEqsMaxwell}
  Z^{P^{ik}}_k=0
\end{align}
corresponding to the subset of Maxwell equations
\[
\nabla\cdot\mathbf{E}=0,\qquad\nabla\times\mathbf{B}-\partial_t\mathbf{E}=0.
\]
These equations (together to the previously found) describe the bundle $G_0\subset G_4\left(T\widetilde{W}_\lambda\right)$; accordingly, we must define the Grassmannian bundle, which can be set in \emph{Reduce} by
\begin{reduce}
  ka(-i) := d a(-i) - xia(-j,-i)*d x(j);
  kp(j,k) := d p(j,k) - xip(-i,j,k)*d x(i);
  kf(-j,-k) := d f(-j,-k) - xif(-i,-j,-k)*d x(i);
  ListPfaff:=index_expand {ka(-i), kp(j,k), kf(-j,-k)};
  Xes:=index_expand {d x(i)};
  NewEDS0:=eds(ListPfaff,Xes);
\end{reduce}
By means of the command \verb|pullback| we could define the EDS $\cI_0$ representing on $G_0$ the restricted Hamiltonian system
\begin{reduce}
  NewEDS:=pullback(NewEDS0,ListSols1.1);
\end{reduce}
This EDS is almost a Linear Pfaffian system; it fails in doing it because it contains a set of $0$-forms; shortly, it is composed by the consistency conditions
\begin{equation}\label{Eq:DerivFP}
  Z^{P^{ij}}_k=\frac{1}{32}g^{ip}g^{jq}Z^{F_{pq}}_k
\end{equation}
which are nothing but the ``derivatives'' of \eqref{Eq:PFuncF}, and the consequences of these equations when applied to the first set of equations of motion \eqref{Eq:FirstEqsMaxwell}. According to the algorithm, it will be necessary to pullback again the EDS to the submanifold $G_1\subset G_0$ defined by these new restrictions; it is achieved by the commands
\begin{reduce}
  PullNewEDS:=pullback(NewEDS,RestNewEDS1.1);
  characters(PullNewEDS);
  involutive(PullNewEDS);
  quasilinear(PullNewEDS);
  TorsionPullNewEDS:=torsion(PullNewEDS);
\end{reduce}
where we measure different aspects of this EDS $\cI_1$. A remarkable fact is that $\cI_1$ has nontrivial torsion and Cartan characters $\{10,9,7,4\}$. This torsion is equivalent to the requirement
\begin{equation}\label{Eq:SecondMaxwel}
  Z^{F_{ij}}_{k}+Z^{F_{ki}}_{j}+Z^{F_{jk}}_{i}=0,
\end{equation}
which are the remaining equations, equivalent in this context to the set of Maxwell homogeneous equations
\[
\nabla\cdot\mathbf{B}=0,\qquad\nabla\times\mathbf{E}+\partial_t\mathbf{B}=0.
\]
These must be eliminated by an additional pullback to a submanifold $G_2$ where these conditions are met; it is performed by the commands
\begin{reduce}
  RestTorsionPullNewEDS:=for each ii in TorsionPullNewEDS join {ii = 0};
  RestTorsionPullNewEDS1:=solve(TorsionPullNewEDS,ListVars1);
  Pull2NewEDS:=pullback(PullNewEDS,RestTorsionPullNewEDS1.1);
  characters(Pull2NewEDS);
  involutive(Pull2NewEDS);
  quasilinear(Pull2NewEDS);
  TorsionPull2NewEDS:=torsion(Pull2NewEDS);
\end{reduce}
yielding to an \emph{involutive EDS} with characters $\{10,9,6,1\}$. The algorithm stops. The full procedure will be shown on the diagram in Figure \ref{Fig:ConstraintMaxwell}.
\begin{figure}
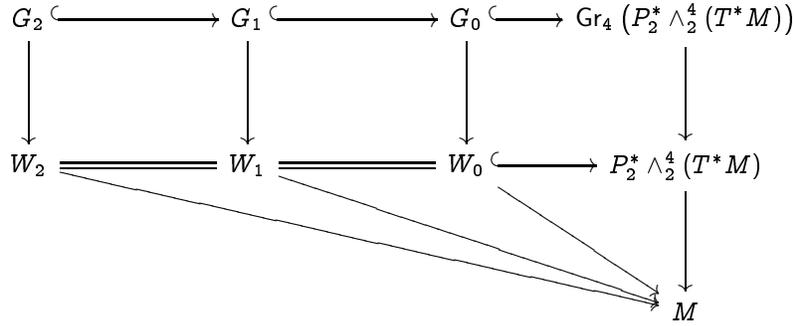

  \centering
  \[
  \begin{diagram}
    \node{G_2}\arrow{e,t,J}{}\arrow{s,l}{}\node{G_1}\arrow{e,t,J}{}\arrow{s,l}{}\node{G_0}\arrow{e,t,J}{}\arrow{s,r}{}\node{\mathsf{Gr}_4\left(P_2^*\wedge^4_2\left(T^*M\right)\right)}\arrow{s,r}{}\\
    \node{W_2}\arrow{e,b,=}{}\arrow{seee,b}{}\node{W_1}\arrow{e,b,=}{}\arrow{see,b}{}\node{W_0}\arrow{e,b,J}{}\arrow{se,b}{}\node{P_2^*\wedge^4_2\left(T^*M\right)}\arrow{s,r}{}\\
    \node[4]{M}
  \end{diagram}
  \]  
  \caption{The constraint structure for Maxwell equations.}
  \label{Fig:ConstraintMaxwell}
\end{figure}
As before $G_0$ is the subset of $\mathsf{Gr}_4\left(P_2^*\wedge^4_2\left(T^*M\right)\right)$ determined by \eqref{Eq:PFuncF}, \eqref{Eq:AllG}, \eqref{Eq:G03} and \eqref{Eq:FirstEqsMaxwell}, namely the Hamilton submanifold; the projected set $W_0$ is the subset of $P_2^*\wedge^4_2\left(T^*M\right)$ locally described by \eqref{Eq:PFuncF}. The set $G_1\subset G_0$ is the zero level set of the $0$-forms in the EDS $\cI_0$ determined by the contact structure on $G_0$, namely \eqref{Eq:DerivFP} and
\[
g^{lk}Z^{F_{ik}}_l=0;
\]
its projection does not produce any further restriction in $W_0$. Finally we need to absorb the torsion \eqref{Eq:SecondMaxwel} of the EDS $\cI_1$ induced by $\cI_0$ on $G_1$; again, no restrictions arise from projection onto $P_2^*\wedge^4_2\left(T^*M\right)$. As we said earlier, the algorithm stops here, because the induced EDS $\cI_2$ is involutive, meaning that the multivector field $Z$ obeying the constraints characterizing $G_2$ is an integrable $4$-vector.

\subsection{Example: EDS with an integrability condition}

\subsubsection{Geometrical setting}

The PDE system
\begin{equation}\label{Eq:IntPDE1}
  \begin{cases}
    u_z+yu_x=0&\\
    u_y=0
  \end{cases}
\end{equation}
can be written as the EDS $\cI\subset\wedge{}^\bullet\mR^4$ generated by the forms
\begin{align*}
  \theta^1&:=\dif u\wedge\dif y\wedge\left(\dif x-y\dif z\right)\\
  \theta^2&:=\dif u\wedge\dif x\wedge\dif z
\end{align*}
with the independence condition $\dif x\wedge\dif y\wedge\dif z$; therefore their solutions are the extremals of the variational problem
\[
\left(\mR^4\rightarrow\mR^3,0,\cI\right).
\]
The following diagram gives the underlying structure characterizing a multisymplectic version of this variational problem:
\[
\begin{diagram}
  \node{\wedge{}^3_2\mR^4}\arrow{e,t}{\text{id}}\arrow[2]{s,l}{\bar\tau}\node{\wedge{}^3_2\mR^4}\arrow[2]{s,r}{\bar\tau}\\
  \\
  \node{\mR^4}\arrow{e,b}{\text{id}}\node{\mR^4}\arrow{e,t}{p}\node{\mR^3}\\
  \node[2]{\left(x,y,z,u\right)}\arrow{e,b,T}{}\node{\left(x,y,z\right)}
\end{diagram}
\]
Thus we will have that
\begin{align*}
  W_{\left(x,y,z,u\right)}&:=0+I\cap\left(\wedge{}^3_2\mR^4\right)_{\left(x,y,z,u\right)}\\
  &=\left\{p_1\theta^1+p_2\theta^2:\left(p_1,p_2\right)\in\mR^2\right\}
\end{align*}
where $I\rightarrow\mR^4$ is a subbundle of $\wedge{}^\bullet\mR^4$ whose sections generate $\cI$.

\subsubsection{Contravariance}\label{Sec:ContravarianceInt1}

The contravariance means that for every $\left(x,y,z\right)\mapsto\left(x,y,z,u\right)$ integral section of $\cI$, there exists a pair of functions $p_1,p_2$ such that
\[
\left(x,y,z\right)\mapsto\left(x,y,z,u,p_1,p_2\right)
\]
is an integral section of the Euler-Lagrange equations of the Lepage-equivalent problem. So, in order to establish this crucial property, it is necessary to take a solution
\[
\left(x,y,z\right)\mapsto u\left(x,y,z\right)
\]
of our PDE system and to found a map
\[
\left(x,y,z\right)\mapsto\left(p_1,p_2\right)
\]
which is integral for the EDS generated by
\[
\gamma:=\dif p_1\wedge\dif y\wedge\left(\dif x-y\dif z\right)+\dif p_2\wedge\dif x\wedge\dif z.
\]
This means that these functions must verify the PDE
\[
\left(p_1\right)_z+y\left(p_1\right)_x+\left(p_2\right)_y=0;
\]
any solution of this equation (for example, $p_1=0,p_2=f\left(x,z\right)$) allows us to prove the desired contravariance.

\subsubsection{Restricted Hamiltonian-like system}

Let us consider $j:W\hookrightarrow\wedge{}^3\mR^4$ the immersion of $W$ in the space of $2$-horizontal $3$-forms; so we define
\[
\Theta_W:=j^*\Theta
\]
where $\Theta\in\Omega^3\left(\wedge^3\mR^4\right)$ is the canonical $3$-form on $\wedge{}^3\mR^4$. The \emph{solutions} of the \emph{restricted Hamiltonian system} $\left(W,\Theta_W\right)$ are the integral sections of the decomposable multivector field $Z:W\rightarrow\wedge{}^3TW$ such that
\begin{equation}\label{Eq:ResHamIntCond}
  Z\lrcorner\dif\Omega=0.
\end{equation}
Let us write $Z=Z_x\wedge Z_y\wedge Z_z$, where
\begin{align*}
  Z_x&:=\partial_x+Z_x^u\partial_u+Z_x^{p_1}\partial_{p_1}+Z_x^{p_2}\partial_{p_2}\\
  Z_y&:=\partial_y+Z_y^u\partial_u+Z_y^{p_1}\partial_{p_1}+Z_y^{p_2}\partial_{p_2}\\
  Z_z&:=\partial_z+Z_z^u\partial_u+Z_z^{p_1}\partial_{p_1}+Z_z^{p_2}\partial_{p_2};
\end{align*}
the system \eqref{Eq:ResHamIntCond} gives us the restrictions
\begin{equation}\label{Eq:ZConditions}
  Z_x^u=-y^{-2}Z_z^u,\quad Z_x^{p_1}=-y^{-2}\left(Z_y^{p_2}+Z_z^{p_1}\right),\quad Z_y^u=0.
\end{equation}
In general, the system \eqref{Eq:ResHamIntCond} admits solutions only in a subset $C\subset W$; even assuming $C$ is a submanifold, it is necessary to ensures us that $Z$ is a true solution, namely
\begin{enumerate}
\item It is tangent to $C$, i.e. that it is a map $Z:C\rightarrow\wedge{}^3C$, and
\item It is integrable on $C$.
\end{enumerate}
The failure of any of these conditions forces us to restrict ourselves to a subset $C_1^{A}\subset C$, where the superindex $A$ indicates if it is the subset where the \emph{tangency condition} is fulfilled (i.e., when $A=T$) or whether we have $A=I$ where the \emph{integrabiltiy condition} is met. From Eqs. \eqref{Eq:ZConditions} it results evident that $C_1^T=W$; although it will be also the case for the integrability conditions, it will be necessary to made some considerations before to reach to this conclusion.

\paragraph{Integrability conditions on $Z$}

The integrability conditions are
\[
\left[Z_i,Z_j\right]\in\left\langle Z_x,Z_y,Z_z\right\rangle 
\]
for all $i,j=x,y,z$. In order to properly work with them, it could be more efficient to describe the subspace $\left\langle Z_x,Z_y,Z_z\right\rangle $ by means of its annihilator $\left\langle Z_x,Z_y,Z_z\right\rangle ^\perp$ spanned by the set of forms
\begin{align}\label{Eq:EDSforZ}
  \begin{split}
    &\Gamma^u:=\dif u-Z^u_x\dif x-Z^u_y\dif y-Z^u_z\dif z,\cr
    &\Gamma^{p_1}:=\dif p_1-Z^{p_1}_x\dif x-Z^{p_1}_y\dif y-Z^{p_1}_z\dif z,\cr
    &\Gamma^{p_2}:=\dif p_2-Z^{p_2}_x\dif x-Z^{p_2}_y\dif y-Z^{p_2}_z\dif z
  \end{split}
\end{align}
constrained by the restrictions \eqref{Eq:ResHamIntCond}. As we know \cite{BCG,CartanBeginners}, a necessary condition for this EDS to be \emph{in involution} is that its \emph{torsion} must be zero; the torsion of this system will be
\[
Z^u_z=0.
\]
Adding this condition to \eqref{Eq:ResHamIntCond} gives us a set of conditions for $Z$ transforming the EDS \eqref{Eq:EDSforZ} into an involutive EDS\footnote{This has to be checked by another method, such as Cartan Test. In the present case this test was carried out by using the package EDS of the computer algebra system REDUCE.}. Thus the multivector $Z:W\rightarrow\wedge{}^3W$ which is tanget, integrable and decomponible is given by the equations
\begin{equation}
  Z_x^u=-y^{-2}Z_z^u,\quad Z_x^{p_1}=-y^{-2}\left(Z_y^{p_2}+Z_z^{p_1}\right),\quad Z_y^u=0=Z^u_z.
\end{equation}
By recalling the meaning of the components of $Z$ we can interpret the integrability condition $Z_x^u=0$; if
\[
\phi:\left(x,y,z\right)\mapsto\left(u(x,y,z),p_1\left(x,y,z\right),p_2\left(x,y,z\right)\right)
\]
is a solution for this system, we see that
\[
Z^m_i=T_{\left(x,y,z\right)}\phi^m\left(\partial_i\right),\qquad i=x,y,z,\quad m=u,p_1,p_2
\]
so $0=Z^u_z=u_z$. It is interesting to note that it is exactly the integrability condition found when playing with the PDE \eqref{Eq:IntPDE1}.

\subsection{Example: An EDS with strong integrability conditions}\label{Sec:EDSStrong}

Let us benefit on the wider class of variational problems we have at our disposal in order to deal with a classical example of PDE system with integrability conditions of higher order; it will allow us to see how the algorithm evolves in dealing with these cases.

\subsubsection{Geometrical preliminaries}

We will try to fit the following variational problem into this scheme, namely
\[
\left(F\rightarrow M,0,\cI_{\text{IC}}\right)
\]
by taking $E:=\mR^4\stackrel{p}{\longrightarrow}M:=\mR^3$ and $F:=J^1p$; the following diagram gives us some idea on what is going on with the maps
\[
\begin{diagram}
  \node{\wedge{}^3_2J^1p}\arrow{e,t}{\text{id}}\arrow[2]{s,l}{\bar\tau}\node{\wedge{}^3_2J^1p}\arrow[2]{s,r}{\bar\tau}\\
  \\
  \node{J^1p}\arrow{e,b}{\text{id}}\node{J^1p}\arrow{e,t}{p_{10}}\node{E}\arrow{e,t}{p}\node{M}\\
  \node[2]{\left(x,y,z,\phi,p,q,r\right)}\arrow{e,b,T}{}\node{\left(x,y,z,\phi\right)}\arrow{e,b,T}{}\node{\left(x,y,z\right)}
\end{diagram}
\]
and the restriction EDS, the system we want to study, is given by
\begin{multline*}
  \cI_{\text{IC}}:=\Big\langle \theta:=\dif\phi-p\dif x-q\dif y-r\dif z,\\
  \Gamma_1:=\dif x\wedge\dif y\wedge\dif r+y\dif y\wedge\dif z\wedge\dif p,\Gamma_2:=\dif x\wedge\dif z\wedge\dif q\Big\rangle _{\text{diff}}.
\end{multline*}
In the previous diagram we use the identifications
\begin{equation}\label{Eq:VertBundleTrivial}
  \left(\wedge{}^3_2J^1p\right)^V=\text{id}^*\left(\wedge{}^3_2J^1p\right)=\wedge{}^3_2J^1p,
\end{equation}
which come from the fact that $V\left(\text{id}\right)=0$, and so $\left(V\left(\text{id}\right)\right)^0=TJ^1p$. 

\subsubsection{Contravariance of the proposed Lepage-equivalent problem}\label{Sec:ContravarianceInt2}

As warned before, it is necessary to ensure the contravariance of these Lepage-equivalent problems before proceed further with the study of its Hamiltonian-like system. By using Proposition \ref{Prop:OursEquivalentOnVars}, it is possible to establish that, besides the pullback of the EDS $\cI_{\text{IC}}$ , the Euler-Lagrange equations for the Lepage-equivalent problem contains the following forms
\[
\gamma_1:=\dif\alpha,\quad\gamma_2:=\alpha\wedge\dif x+y\dif\lambda_1\wedge\dif y\wedge\dif z,\quad\alpha\wedge\dif y+\dif\lambda_2\wedge\dif x\wedge\dif z,\quad\gamma_4:=\alpha\wedge\dif z+\dif\lambda_1\wedge\dif x\wedge\dif y
\]
where $\alpha:=A\dif x\wedge\dif y+B\dif x\wedge\dif z+C\dif y\wedge\dif z$. Given a section $\sigma:\left(x,y,z\right)\rightarrow\left(\phi,p,q,r\right)$, the solutions $\Sigma:\left(x,y,z\right)\mapsto\left(\phi,p,q,r,\alpha,\lambda_1,\lambda_2\right)$ of the EDS $\cJ:=\left\langle \gamma_1,\cdots,\gamma_4\right\rangle $ whose first components are determined by $\sigma$ defines those sections needed to proving the contravariance. We are now in danger to run into a dead end, because it is a new EDS with possibly its own integrability conditions; nevertheless, we do not need any detailed knowledge of this EDS, but only if it has solutions under the previously stated conditions. In fact, it is immediate to find that this EDS is equivalent to the PDE system
\[
\begin{cases}
  C+y\left(\lambda_1\right)_x=0,&\\
  B+\left(\lambda_2\right)_y=0,&\\
  A+\left(\lambda_1\right)_z=0,&\\
  A_z-B_y+C_x=0
\end{cases}
\]
with an easily found solution given by
\[
\left(x,y,z\right)\mapsto\left(\alpha=-f\left(x,z\right)\dif x\wedge\dif z,\lambda_1=0,\lambda_2=f\left(x,z\right)y\right)
\]
for $f$ an smooth function of two variables. Any of these solutions can be used in order to prove the desired contravariance of the proposed Lepage-equivalent problem.

\subsubsection{Restricted Hamiltonian-like system}

If $I_{\text{IC}}\subset\wedge{}^\bullet J^1p$ is the subbundle spanned by a set of generators of $\cI_{\text{IC}}$, the bundle $W\rightarrow J^1p$ determined by these data is given by
\begin{align}
  \left.W\right|_{\left(x,y,z,\phi,p,q,r\right)}&:=\left\{0+\left.I_{\text{IC}}\right|_{\left(x,y,z,\phi,p,q,r\right)}\cap\left(\wedge{}^3_2J^1p\right)^V\right\}\cr
  &=\left.I_{\text{IC}}\right|_{\left(x,y,z,\phi,p,q,r\right)}\cr
  &=\Big\{\left(A\dif x\wedge\dif y+B\dif x\wedge\dif z+C\dif y\wedge\dif z\right)\wedge\theta+\lambda_1\Gamma_1+\lambda_2\Gamma_2:\cr
  &\pushright{A,B,C,\lambda_1,\lambda_2\in\mR\Big\}.}\label{Eq:RestrictedPhaseSpace}
\end{align}
Let us define the $m$-form
\[
\Theta_h:=j^*\Theta
\]
where $j:W\hookrightarrow\wedge{}^3_2J^1p$ is the canonical immersion; thus we can prove the following result.
\begin{prop}[Co- and contravariance of the Lepage-equivalent problem]
  The extremals of the variational problem $\left(W\rightarrow M,\Theta_h,0\right)$, namely, sections $\sigma:M\rightarrow W$ such that
  \[
  \sigma^*\left(X\lrcorner\dif\Theta_h\right)=0\qquad\text{for all }X\in\mathfrak{X}^V\left(W\right)
  \]
  are in a one-to-one correspondence with the extremals of the original variational problem
  \[
  \left(F\rightarrow M,0,\cI_{\text{IC}}\right).
  \]
\end{prop}
Thus our previous discussion allow us to change the problem into a kind of \emph{restricted Hamiltonian system}, namely, to find the integrable $3$-multivectors $X_h\in\wedge{}^3\left(W\right)$ such that
\begin{equation}\label{Eq:RestHamStrong}
  X_h\lrcorner\dif\Theta_h=0,\qquad X_h\lrcorner\left(\dif x\wedge\dif y\wedge\dif y\right)=1.
\end{equation}

\subsubsection{The constraint algorithm}

As we saw before, in order to solve \eqref{Eq:RestHamStrong} it is necessary to construct the Hamilton submanifold $G_0$ in $\mathsf{Gr}_3\left(W\right)$. By using the corresponding definition it appears that this submanifold is defined by
\begin{align*}
  &Z^\phi_x=p,\quad Z^\phi_y=q,\quad Z^\phi_z=r\\
  &yZ^p_x+Z^r_z=0,\quad Z^q_y=0\\
  &Z^C_x-Z^B_y+Z^A_z=0\\
  &yZ^{\lambda_1}_x+C=0,\quad Z^{\lambda_2}_y+B=0,\quad Z^{\lambda_1}_z+A=0.
\end{align*}
The EDS $\cI_0$ induced by the contact structure has the Cartan characters $\left\{7,6,5\right\}$ and torsion spanned by
\[
Z^r_y-Z^q_z,\quad Z^r_x-Z^p_z,\quad Z^q_x-Z^p_y,\quad Z^C_z-yZ^A_x.
\]
The annihilation of first three functions are equivalent to the commutativity of the second order derivatives. The zero torsion locus is the submanifold $G_1\subset G_0$, and the Hamilton Pfaffian induces the EDS $\cI_1$ on it, whose Cartan characters are $\left\{7,5,2\right\}$, and thus not involutive\footnote{Clearly it has no torsion, because it was defined on the zero locus of the torsion of $\cI_0$}. In order to continue the search of an involutive EDS, we need to prolong the EDS $\left(G_1,\cI_1\right)$; it gives rise to the EDS composed by the Grassmann bundle $G_0^{\left(1\right)}:=\mathsf{Gr}_3\left(G_1\right)$ together with its contact structure $\cI_0^{\left(1\right)}$. This operation adds the coordinates
\[
Z^A_{y\alpha},Z^B_{x\alpha},Z^B_{y\alpha},Z^B_{z\alpha},Z^C_{y\alpha},Z^C_{z\alpha},Z^{\lambda_1}_{y\alpha},Z^{\lambda_2}_{x\alpha},Z^{\lambda_2}_{z\alpha},Z^q_{x\alpha},Z^r_{\alpha\beta}
\]
where $\alpha,\beta=x,y,z$. This EDS has Cartan characters $\left\{13,5,2\right\}$ and the torsionless condition translates into
\[
Z^q_{1x}=0,
\]
yielding to the first integrability condition, equivalent to $\phi_{xxy}=0$ on an integral section. The pullback along this new condition gives the EDS $\left(G_1^{\left(1\right)},\cI_1^{\left(1\right)}\right)$ with Cartan characters $\left\{12,5,2\right\}$ and no torsion; it is not involutive, so we need to prolong again. This new prolongation adds the coordinates corresponding to the second derivatives of the components of $Z$, and defines the EDS $\cI_2^{\left(0\right)}$ on $G_2^{\left(0\right)}:=\mathsf{Gr}_3\left(G_1^{\left(1\right)}\right)$; the zero torsion locus for this EDS is given by
\[
\left(Z^r_x\right)_{xz}=0
\]
that on an integral section is equivalent to $\phi_{zxzx}=0$. It is an additional constraint associated to this system; going into the pullback EDS $\left(G_2^{\left(1\right)},\cI_2^{\left(1\right)}\right)$ we see that it has no torsion, although its Cartan characters becomes $\left\{17,6,1\right\}$, and so it is not involutive. Performing an additional prolongation to an EDS $\left(G_3^{\left(0\right)},\cI_3^{\left(0\right)}\right)$, we obtain an involutive EDS with Cartan characters $\left\{22,8,1\right\}$; the algorithm must stops here, and we could present these operations in a diagram:
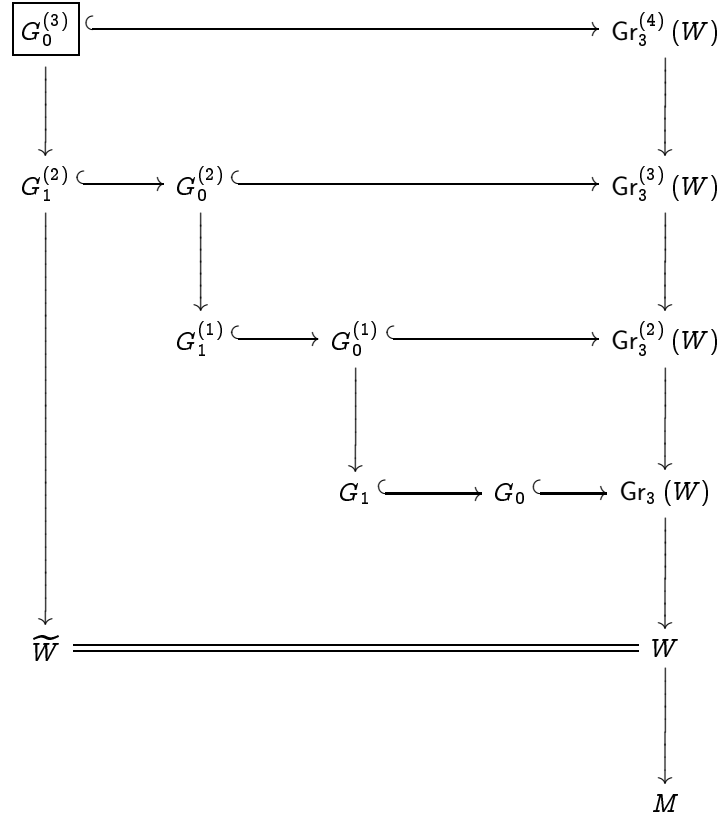
\[
\begin{diagram}
  \node{\boxed{G_0^{\left(3\right)}}}\arrow[4]{e,t,J}{}\arrow{s,l}{}\node[4]{\mathsf{Gr}_3^{\left(4\right)}\left(W\right)}\arrow{s,r}{}\\
  \node{G_1^{\left(2\right)}}\arrow[3]{s,l}{}\arrow{e,b,J}{}\node{G_0^{\left(2\right)}}\arrow[3]{e,t,J}{}\arrow{s,l}{}\node[3]{\mathsf{Gr}_3^{\left(3\right)}\left(W\right)}\arrow{s,r}{}\\
  \node[2]{G_1^{\left(1\right)}}\arrow{e,b,J}{}\node{G_0^{\left(1\right)}}\arrow[2]{e,t,J}{}\arrow{s,l}{}\node[2]{\mathsf{Gr}_3^{\left(2\right)}\left(W\right)}\arrow{s,r}{}\\
  \node[3]{G_1}\arrow{e,b,J}{}\node{G_0}\arrow{e,b,J}{}\node{\mathsf{Gr}_3\left(W\right)}\arrow{s,r}{}\\
  \node{\widetilde{W}}\arrow[4]{e,b,=}{}\node[4]{W}\arrow{s,l}{}\\
  \node[5]{M}
\end{diagram}
\]
where, as before $G_0^{\left(k\right)}=\mathsf{Gr}_3\left(G_{j}^{\left(k-1\right)}\right)$ and
\[
\mathsf{Gr}_3^{\left(k\right)}=\underbrace{\mathsf{Gr}_3\circ\cdots\circ\mathsf{Gr}_3}_{k\text{ times}}.
\]
The fact $\widetilde{W}=W$ encodes the fact that no restriction is imposed on the coordinates
\[
\phi,A,B,C,p,q,r,\lambda_1,\lambda_2;
\]
instead, these constraints are absorbed by the components of $Z$.

\subsection{Example: Classical first order field theory with prolongations}

It is very interesting to note that there exists a classical first order field theory where the application of the Cartan algorithm on its Hamilton system requires to perform prolongations \cite{zbMATH01857369}. The bundle of fields is taken to be
\[
\pi:\mR^6\rightarrow\mR^3:\left(t,x,y,u,v,w\right)\mapsto\left(t,x,y\right)
\]
and its underlying Lagrangian function is given by
\[
L\left(j^1_{\left(x,y,z\right)}\left(u,v,w\right)\right):=\frac{1}{2}\left(u_t^2+yu_x^2\right)+v_yu_y+vw.
\]
The diagram defining the restricted Hamiltonian system is
\[
\begin{diagram}
  \node{\pi_{10}^*\left(\wedge^3_2\mR^6\right)}\arrow{e,l}{}\arrow{s,t}{}\node{\wedge^3_2\mR^6}\arrow{s,r}{}\\
  \node{J^1\pi}\arrow{e,b}{\pi_{10}}\node{\mR^6}\arrow{e,b}{\pi}\node{\mR^3}
\end{diagram}
\]
such that the space for this Hamiltonian system will be composed by the $3$-forms
\[
\widetilde{W}_L:=\left\{\left(L-A^iu_i-B^iv_i-C^iw_i\right)\eta+A^i\eta_i\wedge\dif u+B^i\eta_i\wedge\dif v+C^i\eta_i\wedge\dif w\right\}
\]
on $J^1\pi$, where the indices $i,j$ runs into the values $x,y,z$, $\eta:=\dif t\wedge\dif x\wedge\dif y$ and
\[
\eta_x:=\frac{\partial}{\partial x}\lrcorner\eta,\quad\eta_y:=\frac{\partial}{\partial y}\lrcorner\eta,\quad\eta_t:=\frac{\partial}{\partial t}\lrcorner\eta.
\]
The Hamilton equations read
\[
Z\lrcorner\Omega^L=0,
\]
with $\Omega^L$ the restriction of the canonical $4$-form of $\wedge^3J^1\pi$ to this subbundle of forms, and $Z=Z_t\wedge Z_x\wedge Z_y$ a decomposable $3$-multivector on $\widetilde{W}_L$ with components
\[
Z_i:=\frac{\partial}{\partial x^i}+Z^u_i\frac{\partial}{\partial u}+\cdots+Z^{u_j}_i\frac{\partial}{\partial u_j}+\cdots+Z^{C^j}_i\frac{\partial}{\partial C^j}.
\]
The components of $Z$ can be considered as the coordinates on the fibres of the bundle $\mathsf{Gr}_3\left(\widetilde{W}_L\right)$, and this viewpoint allows us to consider the Hamilton equations as equations defining a submanifold $G_0$ there; these equations turns out to be
\begin{subequations}
  \begin{align}
    &A^t=u_t,\quad A^x=yu_x,\quad A^y=v_y,\quad B^t=B^x=0,\quad B^y=u_y,\quad C^t=C^x=C^y=0,\label{Eq:HamEqsSeilerA}\\
    &Z^{u}_t=u_t, Z^{v}_t=v_t, Z^{w}_t=w_t, Z^{u}_x=u_x, Z^{v}_x=v_x, Z^{w}_x=w_x, Z^{u}_y=u_y, Z^{v}_y=v_y, Z^{w}_y=w_y,\label{Eq:HamEqsSeilerB}\\
    &Z^{A^t}_t+Z^{A^x}_x+Z^{A^y}_y=0,\quad Z^{B^t}_t+Z^{B^x}_x+Z^{B^y}_y=v,\quad Z^{C^t}_t+Z^{C^x}_x+Z^{C^y}_y=w.\label{Eq:HamEqsSeilerC}
  \end{align}
\end{subequations}
The conditions \eqref{Eq:HamEqsSeilerA} restricts the form $\Omega^L$ to be the usual Cartan form for this field theory; in fact, they are equivalent to
\[
A^i=\frac{\partial L}{\partial u_i},\quad B^i=\frac{\partial L}{\partial v_i},\quad C^i=\frac{\partial L}{\partial w_i}.
\]
Additionally, they determine a submanifold $W_0\subset\widetilde{W}_L$ on which the subbundle $G_0$ is fibered. The next line, Eqs. \eqref{Eq:HamEqsSeilerB}, means that the functions $u_i,v_i,w_i$ will be the derivatives of the functions $u,v,w$ on the solutions; just \eqref{Eq:HamEqsSeilerC} are true equations of motion in the sense of classical field theory. This reveals that our approach includes some of the granted relations between quantities as equations of motion.
\newline
The restriction of the contact structure of $\mathsf{Gr}_3\left(\widetilde{W}_L\right)$ to $G_0$ gives rise to an EDS $\cI_0$, which contains a number of $0$-form (i.e. functions) that must be absorbed, as Cartan algorithm dictates; it gives rise to a new set of restrictions to be adopted
\begin{align*}
  &Z^{B^x}_x+Z^{B^y}_y=w,\quad Z^{B^t}_t=Z^{B^t}_x=Z^{B^t}_y=Z^{B^x}_t=Z^{B^x}_x=Z^{B^x}_y=0,\\
  &Z^{C^x}_x+Z^{C^y}_y=w,\quad Z^{C^t}_t=Z^{C^t}_x=Z^{C^t}_y=Z^{C^x}_t=Z^{C^x}_x=Z^{C^x}_y=Z^{C^y}_t=Z^{C^y}_x=Z^{C^y}_y=0,\\
  &Z^{A^x}_x+Z^{A^y}_y+Z^{u_t}_t=0, Z^{A^t}_x-Z^{u_t}_x=0, Z^{A^t}_y-Z^{u_t}_y=0, Z^{A^x}_t-yZ^{u_x}_t=0, Z^{A^x}_x-yZ^{u_x}_x=0,\\
  &Z^{A^y}_t-Z^{v_y}_t=0, Z^{A^y}_x-Z^{v_y}_x=0, Z^{A^y}_y-Z^{v_y}_y=0, Z^{A^x}_y-yZ^{u_x}_y=u_x, \\
  &Z^{B^y}_t-Z^{u_y}_t=0,\quad Z^{B^y}_y-Z^{u_y}_y=0, \quad Z^{B^y}yt-Z^{u_y}_y=0.
\end{align*}
A consequence of these equations is
\begin{equation}
  \label{Eq:FirstIntCondSeiler}
  v=0,
\end{equation}
which becomes an equation of motion; these conditions define a subbundle $G_1\subset G_0$ which is fibered on the submanifold $W_1\subset W_0$. The EDS $\cI_1$ induced on $G_1$ has $0$-forms yielding to the additional constraints
\[
v_t=v_x=v_y=0;
\]
the continuation of the Cartan algorithm forces us to use them in order to define the subbundle $G_2\subset G_1$ fibering onto $W_2\subset W_1$, where $\cI_1$ induces a new EDS $\cI_2$. Again, a set of $0$-forms arise, imposing
\begin{align*}
  &Z^{v_t}_t=Z^{v_t}_x=Z^{v_t}_y=Z^{v_x}_t=Z^{v_x}_x=Z^{v_x}_y=Z^{A^y}_t=Z^{A^y}_x=Z^{A^y}_y=0.
\end{align*}
These are simply the annihilation of the second derivatives of $v$, expressed in terms of the chosen free variables on $G_2$; they determine the subbundle $G_3\subset G_2$ fibered on $W_3=W_2$. The induced EDS $\cI_3$ has no $0$-forms at last; nevertheless, it is not involutive, because the occurrence of torsion, whose absortion, as the Cartan algorithm tells us, yields to the restriction
\[
Z^{w_x}_y-Z^{w_y}_x=Z^{w_x}_t-Z^{w_t}_x=Z^{w_t}_y-Z^{w_y}_t=0, Z^{A^t}_y-Z^{B^y}_t=0, Z^{A^x}_y-yZ^{B^y}_x=u_x, Z^{A^x}_t=yZ^{A^t}_x.
\]
As usual, the first appearance of torsion is related to the equality of crossed second order derivatives; thus it determines a new subbundle $G_4\subset G_3$ with no restriction in its base, i.e. $W_4=W_3=W_2$. The EDS $\cI_4$ induced on $G_4$ by $\cI_3$ has nor $0$-forms neither torsion, but it has Cartan characters $\left\{6,3,1\right\}$, and so \emph{it is not involutive}. The Cartan algorithm says us to prolong this EDS in the search of an equivalent involutive EDS; thus we obtain an EDS $\cI_0^{\left(1\right)}$ on $G_0^{\left(1\right)}\subset\mathsf{Gr}_3^{\left(2\right)}\left(\widetilde{W}_L\right)$, which has torsion. The associated restriction arising from the absortion of this torsion is
\[
2Z^{A^x}_x-2y\left(Z^{A^x}_x\right)_y-y^2Z^{w_t}_t-y^3Z^{w_x}_x=0,
\]
thus determining a subbundle $G_1^{\left(1\right)}\subset G_0^{\left(1\right)}$ fibered on $W_4=W_3=W_2$; on a solution it is equivalent to the PDE
\[
2u_{xxy}+yw_{xx}+w_{tt}=0
\]
if is taken into account that in such case $Z^{A^x}_x=\left(A^x\right)_x,A^x=yu_x$ and every subindex means partial derivative. The EDS $\cI_1^{\left(1\right)}$ obtained by pullback along this constraint has no $0$-forms but it presents torsion, giving rise to the additional constraint
\[
3Z^{w_x}_x+\left(Z^{w_t}_t\right)_y+y\left(Z^{w_x}_x\right)_y=0;
\]
it yields to the PDE
\[
w_{ytt}+yw_{xtt}+3w_{xx}=0
\]
by using that on solutions we have the identifications $Z^{w_x}_x=w_{xx}, Z^{w_t}_t=w_{tt}$ and, as before, the subindices mean partial derivatives everywhere. Anyway, from this constraint emerges the subbundle $G_2^{\left(1\right)}\subset G_1^{\left(1\right)}$ without further restrictions on the base; on it appears the EDS $\cI_2^{\left(1\right)}$ without $0$-forms, with zero torsion, but with the Cartan characters $\left\{9,3,0\right\}$, and therefore not involutive. Thus it is necessary to perform a prolongation to an EDS $\cI^{\left(2\right)}_0$ induced by the contact structure on $G_0^{\left(2\right)}\subset\mathsf{Gr}_3^{\left(3\right)}\left(\widetilde{W}_L\right)$, which has no $0$-forms but nonzero torsion, determining the restriction
\[
2\left(Z^{A^t}_x\right)_{tx}+y\left(Z^{w_t}_t\right)_{tt}+2y^2\left(Z^{w_t}_t\right)_{xx}+y^3\left(Z^{w_x}_x\right)_{xx}=0.
\]
This is the last of the equations (the fourth order equation) mentioned in \cite{zbMATH01857369}. The zero torsion locus is a subbundle $G_1^{\left(2\right)}\subset G_0^{\left(2\right)}$ fibered onto $W_4$, and the induced EDS $\cI_1^{\left(2\right)}$ has neither torsion nor $0$-forms, and its Cartan characters become $\left\{11,2,0\right\}$, meaning that it is an involutive EDS; so the algorithm must stops here. As before, we can fit all these subbundles in a diagram, as shows Figure \ref{Fig:SeilerLagDiagram}.
\begin{figure}
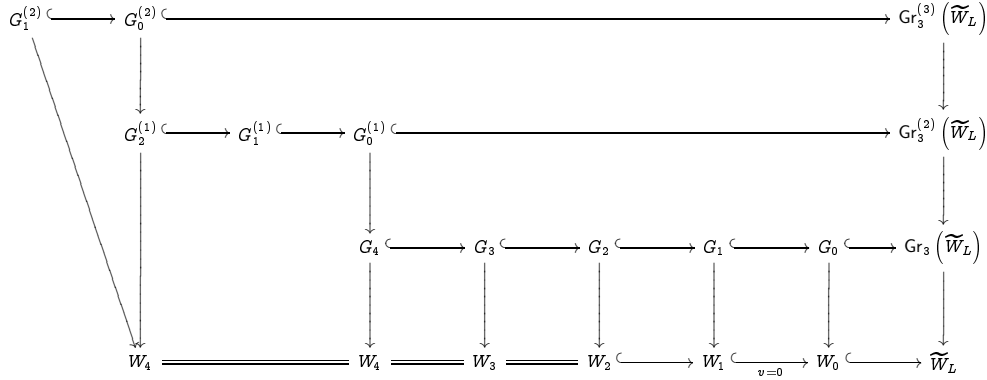

  \begin{adjustbox}{max size={.95\textwidth}{.8\textheight}}
    $\displaystyle
    \begin{diagram}
      \node{G_1^{\left(2\right)}}\arrow{e,b,J}{}\arrow{ssse,b}{}\node{G_0^{\left(2\right)}}\arrow{s,r}{}\arrow[7]{e,t,J}{}\node[7]{\mathsf{Gr}_3^{\left(3\right)}\left(\widetilde{W}_L\right)}\arrow{s,r}{}\\
      \node[2]{G^{\left(1\right)}_2}\arrow[2]{s,r}{}\arrow{e,t,J}{}\node{G^{\left(1\right)}_1}\arrow{e,t,J}{}\node{G^{\left(1\right)}_0}\arrow[5]{e,t,J}{}\arrow{s,r}{}\node[5]{\mathsf{Gr}_3^{\left(2\right)}\left(\widetilde{W}_L\right)}\arrow{s,r}{}\\
      \node[4]{G_4}\arrow{e,t,J}{}\arrow{s,l}{}\node{G_3}\arrow{s,l}{}\arrow{e,t,J}{}\node{G_2}\arrow{s,l}{}\arrow{e,t,J}{}\node{G_1}\arrow{s,l}{}\arrow{e,t,J}{}\node{G_0}\arrow{s,l}{}\arrow{e,t,J}{}\node{\mathsf{Gr}_3\left(\widetilde{W}_L\right)}\arrow{s,l}{}\\
      \node[2]{W_4}\arrow[2]{e,b,=}{}\node[2]{W_4}\arrow{e,b,=}{}\node{W_3}\arrow{e,b,=}{}\node{W_2}\arrow{e,b,J}{}\node{W_1}\arrow{e,b,J}{v=0}\node{W_0}\arrow{e,b,J}{}\node{\widetilde{W}_L}
    \end{diagram}
    $
  \end{adjustbox}
  \caption{The full structure of the Cartan constraints of Lagrangian $L$.}
  \label{Fig:SeilerLagDiagram}
\end{figure}

\subsection{Example: An affine Lagrangian density}
In \cite{DeLeon2005839} an example based on a field theory with an affine Lagrangian is discussed, in order to show the main features of its constraint algorithm. It is our understanding that dealing with this example from the viewpoint of our own algorithm could be useful for exploring its behaviour in a controlled environment.

\subsubsection{Introduction}

The configuration bundle is in this case $\pi:\mR^4\rightarrow\mR^2$ and so $J^1\pi=\mR^4\times\mR^8$ with global coordinates $\left(x^i,y^j,v^i_j\right)$ where $i,j=1,2$. On $J^1\pi$ we consider the $2$-form
\[
\alpha:=y^1y^2\dif x^1\wedge\dif x^2-x^2y^1\dif y^1\wedge\dif x^1-x^2y^2\dif y^2\wedge\dif x^1
\]
defining the (pre)symplectic restricted Hamiltonian equations
\begin{equation}\label{Eq:AffineHam}
  Z\lrcorner\dif\alpha=0,\qquad Z\lrcorner\dif x^1\wedge\dif x^2=1
\end{equation}
for some $Z\in\mathfrak{X}^2\left(J^1\pi\right)$. As we know, it is the right setting for the use of the algorithm developed above.

\subsubsection{Resolution}

According to our method, we need to pullback the canonical contact structure of $\pi_{2,0}:G_2\left(T\left(J^1\pi\right),\dif x^1\wedge\dif x^2\right)\rightarrow\mR^4$ to the subset 
\[
G_0:=\left\{E_{\left(x^i,y^i,v^i_j\right)}:v\wedge w\lrcorner\left(\left.\dif\alpha\right|_{\left(x^i,y^i,v^i_j\right)}\right)=0\text{ for some }v,w\text{ basis of }E\right\}.
\]
By taking into account that
\[
G_2\left(T\left(J^1\pi\right),\dif x^1\wedge\dif x^2\right)=\mR^4\times\mR^8\times\mR^{20}
\]
with global coordinates $\left(x^1,y^j,v^i_j,Y^j_k,V^i_{jk}\right)$, the set $G_0$ is described by the formulas
\[
y^1-y^2=0,\qquad\left(y^1-y^2\right)\left(Y^1_1-Y^2_1\right)=0,\qquad\left(y^1-y^2\right)\left(Y^1_2-Y^2_2\right)=0;
\]
then $G_0$ is the set of points of $G_2\left(T\left(J^1\pi\right),\dif x^1\wedge\dif x^2\right)$ such that $y^1=y^2$, and pulling back the contact structure on it we obtain the EDS $\cI_0$ generated by the $0$-forms 
\[
M:=Y^1_1-Y^2_1,\qquad N:=Y^1_2-Y^2_2
\]
and the set of $1$-forms
\[
\left\{\dif v^i_j-v^i_{jk}\dif x^k,\dif y^2-y^2_i\dif x^i\right\}.
\]
The level zero set of the functions $M,N\in C^\infty\left(G_2\left(T\left(J^1\pi\right),\dif x^1\wedge\dif x^2\right)\right)$ define a submanifold $G_1$ where we must restrict to, so the solutions of \eqref{Eq:AffineHam} are the elements of the fibers of the Grassmann bundle living into $G_1$; the EDS $\cI_1$ gives the geometrical interpretation of the components of these solutions as derivatives of the corresponding dependent variables, and the constraints induced by the requirements of tangency and integrability are recovered as the equations describing the set $E_1:=\pi_{2,0}\left(G_1\right)$. As before, these structures fit in the following diagram
\[
\begin{diagram}
  \node{G_1}\arrow{s,l}{}\arrow{e,t,J}{}\node{G_2\left(T\left(J^1\pi\right),\dif x^1\wedge\dif x^2\right)}\arrow{s,r}{\pi_{2,1}}\\
  \node{K_1}\arrow{s,l}{}\arrow{e,t,J}{}\node{J^1\pi}\arrow{s,r}{\pi_{1,0}}\\
  \node{E_1}\arrow{se,b}{}\arrow{e,t,J}{}\node{E}\arrow{s,r}{\pi}\\
  \node[2]{\mR^2}
\end{diagram}
\]
In short, we will have solutions for \eqref{Eq:AffineHam} only when $y^1=y^2$, and the corresponding $Z$ must verify that $Y_i^1=Y_i^2,i=1,2$.

\subsection{Example: Lagrangian with integrability condition}

Let us discuss an example from \cite{zbMATH00177140}, where a first order (singular, toy) Lagrangian has first order integrability conditions, in order to see how our scheme works in the search of sufficient conditions for the existence of solutions of a field theory. From this work we obtain the Lagrangian definition
\[
L:=u_x\left(w_x+v_y\right)+yw^2
\]
together with the definitions $p:E:=\mR^5\rightarrow\mR^2=:M$ and the diagram
\[
\begin{diagram}
  \node{p_{10}^*\left(\wedge{}^2_2E\right)}\arrow{e,t}{}\arrow[2]{s,l}{\bar\tau}\node{\wedge{}^2_2E}\arrow[2]{s,r}{\bar\tau}\\
  \\
  \node{J^1p}\arrow{e,b}{p_{10}}\node{E}\arrow{e,t}{p}\node{M}\\
  \node[2]{\left(x,y,u,v,w\right)}\arrow{e,b,T}{}\node{\left(x,y\right)}
\end{diagram}
\]
The coordinates on $J^1p$ are the canonical coordinates $\left(x,y,u,v,w,u_x,v_x,w_x,u_y,v_y,w_y\right)$. By using the prescriptions developed above, we construct the submanifold $\widetilde{W}\subset\wedge^\bullet J^1p$ such that
\begin{align*}
  \widetilde{W}&:=L\eta+I_{\text{con}}\cap\left(\wedge^2_2E\right)^V\\
  &=\left\{L\eta+\alpha\wedge\theta^u+\beta\wedge\theta^v+\gamma\wedge\theta^w:\alpha,\beta,\gamma\in\wedge^1M\right\}
\end{align*}
where $\eta:=\dif x\wedge\dif y$, $I_{\text{con}}$ is the subbundle of $\wedge^\bullet J^1p$ whose sections generate $\cI_{\text{con}}$ and
\begin{align*}
  \theta^u&:=\dif u-u_x\dif x-u_y\dif y\\
  \theta^v&:=\dif v-v_x\dif x-v_y\dif y\\
  \theta^w&:=\dif w-w_x\dif x-w_y\dif y;
\end{align*}
additionally, formula \eqref{Eq:VertBundleTrivial} was used in order to symplify the calculations involved in this definition. The canonical $2$-form on $\wedge^2J^1p$ induces a $2$-form $\Theta_L$ on $\widetilde{W}$, and the Hamiltonian-like system yields to the following problem: To find a $2$-multivector decomposable and integrable $Z$ such that
\begin{equation}\label{Eq:HamiltonSaunders}
  Z\lrcorner\dif\Theta_L=0,\qquad Z\lrcorner\eta=1.
\end{equation}
We need to parametrize the spaces $T^*M$ appearing into our description of $\widetilde{W}$; let us use
\begin{align*}
  \alpha:=p\dif x+q\dif y,\qquad\beta:=r\dif x+s\dif y,\qquad\gamma:=m\dif x+n\dif y.
\end{align*}
So if we write down $Z:=Z_x\wedge Z_y$ where
\[
Z_x:=\frac{\partial}{\partial x}+Z^u_x\frac{\partial}{\partial u}+\cdots+Z^n_x\frac{\partial}{\partial n}
\]
and similar for $y$, the Hamilton equations \eqref{Eq:HamiltonSaunders} yields to
\[
Z^A_i=u^A_i
\]
for $A=u,v,w,i=x,y$, and additionally
\begin{align}
  &Z^p_y-Z^q_x=0,\quad Z^r_y-Z^s_x=0,\quad Z^m_y-Z^n_x=2yw\notag\\
  &u_x=-n,\quad w_x=-\left(v_y+q\right),\quad p=s=m=0,\quad r=-n.\label{Eq:W21J21Lagrangian}
\end{align}
The contact structure on $\mathsf{Gr}_2\left(\widetilde{W}\right)$ induces the EDS $\cI_0$ on the submanifold $G_0$ determined by these equations; it results that this EDS contains a bunch of $0$-forms, and its elimination yields to the constraints
\begin{align*}
  &Z^q_y+Z^{v_y}_y+Z^{w_x}_y=0,\quad Z^q_x+Z^{v_y}_x+Z^{w_x}_x=0\\
  &Z^n_y+Z^{u_x}_y=0,\quad Z^n_y+Z^{u_x}_y=0\\
  &Z^p_x=Z^q_x=Z^m_x=Z^s_y=Z^s_x=0\\
  &Z^s_x+Z^n_y=0,\quad Z^n_y+Z^r_x=0\\
  &Z^n_x+2yw=0.
\end{align*}
These constraints gives rise to a new submanifold $G_1$, where the EDS $\cI_1$ is induced; by construction, it has no $0$-forms, but it presents torsion, whose annihilation imposes the additional constraints
\[
Z^{w_y}_x-Z^{w_x}_y=0,\quad Z^{w_x}_x+Z^{v_x}_y=0,\quad Z^{u_y}_x=0,\quad yw_y+w=0.
\]
The last constraint is the first order constraint found in \cite{zbMATH00177140}, the remaining fix some free components of the ``Hamiltonian'' $2$-vector $Z$; thus a submanifold $G_2$ is found, with EDS $\cI_2$. This EDS contains $0$-forms inducing the constraints
\[
yZ^{w_y}_y+2w_y,\qquad yZ^{w_y}_x-v_y-q=0;
\]
under the assumption $y\not=0$, they determine a submanifold $G_3$ with EDS $\cI_3$, which has Cartan character $\left\{4,0\right\}$, and is an involutive EDS. The diagram of the algorithm becomes
\[
\begin{diagram}
  \node{G_3}\arrow{e,t,J}{\left(y\not=0\right)}\arrow{s,l}{}\node{G_2}\arrow{e,t,J}{}\arrow{s,l}{}\node{G_1}\arrow{e,t,J}{}\arrow{s,l}{}\node{G_0}\arrow{e,t,J}{}\arrow{s,l}{}\node{\mathsf{Gr}_2\left(\widetilde{W}\right)}\arrow{s,l}{}\\
  \node{W_3}\arrow{e,t,=}{}\arrow{se,l}{}\node{W_2}\arrow{e,t,J}{A}\arrow{s,l}{}\node{W_1}\arrow{e,t,=}{}\arrow{see,l}{}\node{W_0}\arrow{e,t,J}{B}\arrow{se,l}{}\node{\widetilde{W}}\arrow{s,l}{}\\
  \node[2]{J_2}\arrow[3]{e,t,J}{C}\node[3]{J^1p}
\end{diagram}
\]
The inclusion $A$ is determined by the constraint $yw_y+w=0$ that is projectable and induces the map $C$; the map $B$ arises from \eqref{Eq:W21J21Lagrangian}.

\appendix

\section{Exterior differential systems and involution}
\label{App:EDSInvolution}
There are two operations needing some clarifications: The \emph{involution} issue and the \emph{prolongation procedure}. Without going into details, the idea is to use Cauchy-Kovalevskaia (CK) theorem in order to find sufficient conditions for existence of solutions of the PDE system underlying an EDS. Although this requirement is forcing us to work in the real-analytic realm, it is interesting to note that the EDS describing physical theories are of polynomial nature, and so well covered by these methods. The building blocks from which construct the solutions will be the so called \emph{integral elements}, which are the planes in the tangent space of the manifold where the EDS lives, on which every form in the EDS annihilates. The integral elements forms a set in the \emph{Grassmann bundle} $p_{1,0}:G_m\left(TW,p^*\eta\right)\rightarrow W$, whose fiber on $w\in W$ are made of the $m$-dimensional subspaces $E$ of $T_wW$ such that $\left.p^*\eta\right|_E\not=0$; an integral element is called \emph{ordinary} if the connected components through it is an smooth submanifold of the Grassmann bundle. Addtionally, because we are trying to proceed by dividing the given EDS into a sequence of PDE systems, each fitting the hypothesis of CK theorem, it is necessary to ensures the smoothness of the posible extensions of a given integral element. The ordinary integral elements where it happened are called \emph{regular}, and for them it is true the so called \emph{Cartan-Kähler theorem}; thus we ahve the following definition
\begin{defc}[Involutive EDS]
We say that an EDS with independence condition $\left(\cI,\eta\right)$ on $W$ is \emph{in involution} if and only if there exists a regular element on every $w\in W$.
\end{defc}
Thus as a corollary of the Cartan-Kähler theorem, through every point of the base manifold of an involuvite EDS passes an integral manifold of this EDS. Although the regularity condition is hard to be applied directly, there exists a test devised by Cartan which reduces it to essentially linear algebra manipulations (still hard to be carried out at hands, but suitable to be worked by a computer.) This is achieved by defining the so called \emph{characters of the EDS} associated to an integral element $E$: considering a flag of integral elements 
\[
0=E_0\subset\cdots\subset E_{m-1}\subset E_m=E
\]
whose top element is the given integral element $E$, and looking for the codimensions $c_k$ of their \emph{polar spaces} $H\left(E_k\right),k=0,\cdots,m-1$, namely, the vector space composed by those directions in which $E_k$ can be enlarged to a larger integral element.
\begin{thm}[E. Cartan's involutivity test]
  The set of $m$-integral elements for $\cI$, $V_m\left(\cI\right)\subset\mathsf{Gr}_m\left(W\right)$ is an smooth submanifold in a neighborhood of $E$ if and only if it is the top element of a flag of regular elements such that
  \[
  \mathop{\text{codim}}{V_m\left(\cI\right)}=c_0+\cdots+c_{m-1}.
  \]
\end{thm}
These characters asociated to an integral element $E$ are related to Cartan characters we refer into the text by formulas
\begin{align*}
  s_0&=c_0\\
  s_k&=c_k-c_{k-1},\qquad 1\leq k\leq m-1\\
  s_m&=\mathop{\text{codim}}{E}-c_{m-1},
\end{align*}
and so the Cartan's test reads
\[
\mathop{\text{dim}}{V_m\left(\cI\right)}-\mathop{\text{dim}}{W}=s_1+2s_2+\cdots+ms_m.
\]

\section{Extended Hamiltonian systems}\label{App:EXtendedHamSys}

Let us indicate the way in which the tools developed in the present article allows us to construct not only restricted Hamiltonian system, but extended system too.

\subsection{Classical variational problems}

Let us turn our attention to extended Hamiltonian systems in the classical case. The idea is to consider $W_{L\eta}$ as the image set of a section $h:\pi_{10}^*\left(\wedge^m_2E/\wedge^m_1E\right)\rightarrow\pi_{10}^*\left(\wedge^m_2E\right)$, given locally as
\[
h\left(x^k,u^A,u^A_k,p_A^k\dif u^A\wedge\eta_k\right)=\left(x^k,u^A,u^A_k,p_A^k\dif u^A\wedge\eta_k+\left(L-p_B^lu^B_l\right)\eta\right).
\]
By mimicking \cite{EcheverriaEnriquez:2005ht}, let $\alpha\in\Omega^1\left(\pi_{10}^*\left(\wedge^m_2E\right)\right)$ be a $1$-form such that
\begin{itemize}
\item $\displaystyle\dif\alpha=0$,
\item $\displaystyle\left.\alpha\right|W_{L\eta}=0$, and
\item in the local coordinates $\left(x^k,u^A,u^A_k,p_B^l,p\right)$ on $\pi_{10}^*\left(\wedge^m_2E\right)$ it verifies
  \[
  \frac{\partial}{\partial p}\lrcorner\alpha=1.
  \]
\end{itemize}
In our framework a form locally performing this magic reads
\[
\alpha:=\dif\left(p+p_B^lu_l^B-L\right);
\]
then we have the following result, which is nothing but \cite[Thm.\ 5, p.\ 14]{EcheverriaEnriquez:2005ht} translated to the present context.
\begin{theorem}\label{Thm:Thm5RomanRoy}
  Let $\imath_{L\eta}:W_{L\eta}\hookrightarrow\pi_{10}^*\left(\wedge^m_2E\right)$ be the natural immersion. If $Z^m\in\pi_{10}^*\left(\wedge^m_2E\right)$ is a decomposable $m$-vector solution of the extended Hamiltonian system
  \begin{equation}\label{Eq:ExtendedSystem}
    Z^m\lrcorner\dif\Theta_{10}=\alpha,\qquad Z^m\lrcorner\eta=1.
  \end{equation}
  then there exists a decomposable solution of the restricted system
  \[
  X^m\lrcorner\dif\Theta_{L\eta}=0,\qquad X^m\lrcorner\eta=1
  \]
  which is $\imath_{L\eta}$-related to $Z^m$.
\end{theorem}
\begin{proof}
  Note that $W_{L\eta}$ is a codimension $1$ integral submanifold for the distribution $\cD_\alpha:=\left\langle \alpha\right\rangle ^0$ in $\pi_{10}^*\left(\wedge^m_2E\right)$; now if $Z^m=Z_1\wedge\cdots\wedge Z_m$ we know that $Z^m$ is tangent to an integral submanifold $S$ of $\cD_\alpha$ if and only if every $Z_i$ does, and if $S$ has codimension $1$, this condition translates into
  \[
  \imath_S^*\left(Z_i\lrcorner\alpha\right)=0\qquad\text{for every }i=1,\cdots,m
  \]
  for $\imath_S:S\hookrightarrow\pi_{10}^*\left(\wedge^m_2E\right)$ the canonical immersion. But
  \[
  Z_i\lrcorner\alpha=Z_i\lrcorner Z^m\lrcorner\dif\Theta_{10}=Z_i\wedge\left(Z_1\wedge\cdots\wedge Z_m\right)\lrcorner\dif\Theta_{10}=0,
  \]
  so $Z^m$ is tangent to $W_{L\eta}$; let $X^m=T\imath_{L\eta}Z^m$. Then the restriction of Equation \eqref{Eq:ExtendedSystem} to $W_{L\eta}$ does the magic.
\end{proof}

\subsection{Extended Hamiltonian systems in examples}

Let us see how a generalization of these considerations can be formulated in order to work with some of the examples discussed above.

\subsubsection{Extended Hamiltonian system for Maxwell equations}
Let us return to the results stated in Proposition \ref{Prop:WinMaxwellEqs}. Translated into the corresponding pullback bundle, it means that in order to find an extended formulation of this Hamiltonian system, it will be necessary to take into account several ``Hamilton forms'' instead of the unique form $\alpha$ defined previously. However, we could figure out how to deal with these collection of forms by requesting that an analogous of Theorem \ref{Thm:Thm5RomanRoy} holds in this new context. Let us suppose that we have performed the passage from $\left(\wedge^m_2J^1\pi\right)^V$ to $P_2^*\left(\wedge^m_2\left(T^*M\right)\right)$ and that the set $W_\lambda$ can be written as follows
  \[
  W_\lambda=M_\lambda\cap\mathop{\text{Im}}{h}
  \]
  for $h:P_2^*\left(\wedge^m_2\left(T^*M\right)/\wedge^m_1\left(T^*M\right)\right)\rightarrow P_2^*\left(\wedge^m_2\left(T^*M\right)\right)$ a section. Thus we observe that:
\begin{itemize}
\item The Hamiltonian function $E$ is simply determined by searching of a local expression for the section $h$; in this case we have that
  \[
  h:=p-\left[F^{ij}-\frac{\left(-1\right)^{m+1}}{2}q^{ij}\right]F_{ij}.
  \]
\item The remaining restrictions define $M_\lambda\subset P_2^*\left(\wedge^m_2\left(T^*M\right)\right)$, so they become in this example $q^{ij}+q^{ji}=0$.
\end{itemize}

\subsubsection{The extended Hamiltonian system in Example of Section \ref{Sec:EDSStrong}}

As before, we could use the embedding $W\subset\wedge{}^3_2J^1p$ in order to define the extended version for this system: It is only necessary to find a section $h$ of the bundle $\mu:\wedge{}^3_2J^1p\rightarrow\wedge{}^3_2J^1p/\wedge{}^3_1J^1p$ such that $W=\text{Im}h$. By using the expression \eqref{Eq:RestrictedPhaseSpace} for $W$ we will have that any element $u$ in $W$ could be written as
\begin{multline*}
  u=\left(-a r+b q-c p\right)\eta+\lambda_1 y \dif\,p\wedge\dif\,y\wedge\dif\,z +A\dif\phi\wedge\dif\,x\wedge\dif\,y+B\dif\phi\wedge\dif\,x\wedge\dif\,z +\\
  +C\dif\phi\wedge\dif\,y\wedge\dif\,z +\lambda_2\dif\,q\wedge\dif\,x\wedge\dif\,z+\lambda_1\dif\,r\wedge\dif\,x\wedge\dif\,y
\end{multline*}
for $\eta:=\dif x\wedge\dif y\wedge\dif z$; the coordinates $v=\left(x,y,z,p,q,r,A,B,C,\pi\right)$ if and only if
\begin{multline*}
  v=\pi\eta+\lambda_1 y \dif\,p\wedge\dif\,y\wedge\dif\,z +A\dif\phi\wedge\dif\,x\wedge\dif\,y+B\dif\phi\wedge\dif\,x\wedge\dif\,z +\\
  +C\dif\phi\wedge\dif\,y\wedge\dif\,z +\lambda_2\dif\,q\wedge\dif\,x\wedge\dif\,z+\lambda_1\dif\,r\wedge\dif\,x\wedge\dif\,y
\end{multline*}
means that $\pi=-a r+b q-c p$ defines $h$.
\begin{defc}[Extended Hamiltonian system for $\left(F\rightarrow M,0,\cI_{\text{IC}}\right)$]
  Let $\sigma\in\Omega^1\left(\wedge{}^3_2J^1p\right)$ be
  \[
  \sigma:=\dif\left(-a r+b q-c p-\pi\right).
  \]
  The triple $\left(\wedge{}^3_2J^1p,\dif\Theta,\dif\sigma\right)$ is an \emph{extended Hamiltonian system} in the sense of \cite{EcheverriaEnriquez:2005ht}.
\end{defc}

\bibliographystyle{plainnat}
%\bibliography{/home/santi/Latex/Tesis/Hugo/ConstraintsyOtras/ArchivosTesis/iniciotesis}
%\bibliography{/home/santi/Dropbox/Trabajos/Bibliography/iniciotesis}

\end{document}